\documentclass[12pt]{article}
\usepackage{setspace}

\usepackage{enumerate}
\usepackage{graphicx}
\usepackage[margin=1in]{geometry}
\usepackage{pgf}
\usepackage{tikz}
\usepackage{pgfplots}
\usepackage{caption}
\usepackage{subcaption}
\usepackage{amsmath}
\usepackage{amsfonts}
\usepackage{amssymb}
\usepackage{amsthm}
\usepackage{mathrsfs}
\usepackage{natbib} 
\usepackage{pgfplots}
\pgfplotsset{compat=1.18}
\usepackage{tocbibind}
\usepackage{subcaption}
\usepackage{tikz}
\usepackage{pgfplots}
 
\usepgfplotslibrary{groupplots}
\pgfplotsset{compat=1.18}
\usepackage[toc,page]{appendix}
\usepackage[colorlinks = true,
            linkcolor = blue,
            urlcolor  = blue,
            citecolor = blue,
            anchorcolor = blue]{hyperref}

\usepackage{verbatim}

\usetikzlibrary{decorations.markings}

\newtheorem{lemma}{Lemma}
\newtheorem{prop}{Proposition}

\newtheorem{example}{Example}
\newtheorem{definition}{Definition}

\newlength\tindent
\usepackage{multirow}
\pgfplotstableread[col sep=comma, header=true]{expost_data.csv}\expostdata
\usetikzlibrary{positioning}
\begin{document}

\title{The Depth and Reach of Exploitation: Contracting with Endogenously Naive Consumers\thanks{We thank Roland B{\'e}nabou and Alex Jakobsen for insightful feedback. Earlier versions of this work were presented under the title ``Contract Design with Endogenous Cognition" at UC San Diego, Princeton University, UNSW, UTS, The University of Queensland, Monash University, and the University of Sydney. We thank seminar participants at those institutions for comments on that earlier work.}}
	
	\author{Benjamin Balzer\thanks{University of Technology Sydney (benjamin.balzer@uts.edu.au)} \quad \& \quad Benjamin Young\thanks{University of Technology Sydney (benjamin.young@uts.edu.au)}}
	\date{\today}  
	
	\maketitle

\begin{abstract}
   Consumers can invest resources to understand and avoid their behavioral mistakes, and their incentives to do so depend on the market consequences of remaining naive. We incorporate this feedback between consumers' cognitive states and market outcomes into a general contracting model. Firms face a trade-off between the depth and reach of exploitation: deeper exploitation raises profit from a naive consumer but induces greater cognitive investment, promoting sophistication and shrinking the exploitable consumer base. This trade-off disciplines exploitation and can cause policies that benefit consumers when cognition is fixed to backfire when cognition is endogenous. 
\end{abstract}

 \section{Introduction}\label{sec:intro}

Consumers often make systematic errors. How markets operate in the presence of such naive consumers and their sophisticated counterparts is an important question for both economists and policymakers.\footnote{Systematic consumer errors have been documented across numerous economic settings. Particularly relevant for our application, \citet{DellaVigna2006} document mistaken expectations about future gym attendance, \citet{Meier2010} link present-biased preferences to credit-card borrowing, and \citet{Augenblick2015} provide experimental evidence of present bias in real-effort choices. Behavioral science has also become an important component of policymaking, with governments establishing behavioral-insights teams to inform the design, implementation, and evaluation of policies intended to improve consumer decision-making \citep{OECDBehaviouralScience}.} Moreover, the incidence of naiveté and sophistication appears to vary across individuals and economic contexts.\footnote{Empirical evidence points to systematic variation in consumers' awareness and attention with experience and economic context. In a real-effort experiment, \citet{LeYaouanqSchwardmann2022} find that participants initially overestimate their future effort but become less naive after observing their own behavior, with learning carrying across the related task environments considered in the experiment. In payday lending, where mistakes can be costly and repeated, \citet{AllcottEtAl2022} find that only the most inexperienced quartile of borrowers underpredicts future borrowing, whereas the remaining borrowers predict it correctly on average. Using nationally representative data on ideal, predicted, and realized body weight, \citet{CobbClarkEtAl2024} find that sophisticated individuals exhibit larger deviations from ideal weight than naive individuals, consistent with self-control problems becoming easier to recognize the more severe they are. Conversely, in subscription markets, where individual charges are small, \citet{EinavEtAl2025} document that consumers systematically fail to cancel subscriptions they no longer value, and argue that the cost of inattention is sufficiently low that consumers do not find it worthwhile to monitor these decisions closely.} Yet existing models typically treat the distribution of these cognitive types as fixed, allowing consumers' cognitive states to shape firms' conduct while suppressing the reverse relationship. This one-way treatment is restrictive since consumers can devote resources to understanding their mistakes, and the benefit of doing so depends on the market consequences of remaining naive.\footnote{The broader principle that bounded cognition and the errors it produces respond to economic incentives is well established. \citet{larbi2022} model reasoning depth as a cost--benefit choice and provide experimental evidence that it responds to payoffs. \citet{Enke2021} find that very high stakes increase cognitive effort, albeit with only mild improvements in decision quality; \citet{Zimmermann2020} shows that incentives for belief accuracy mitigate motivated reasoning; and \citet{FehrFinkJack2022} find that subjects are less likely to reject mutually beneficial trades when forgone gains from trade become more consequential. A similar cost--benefit logic appears in the rational-learning models discussed by \citet{Lusardi2014}, which treat financial knowledge as an endogenous investment in human capital.} 
Market interaction may therefore not only be shaped by the distribution of cognitive types but also shape that distribution.  \\

In this paper, we formalize the feedback between market interaction and cognition by modeling naivet\'e and sophistication as the outcomes of a process of costly cognitive investment.\footnote{Becoming sophisticated requires attention, effort, and time, but allows consumers to avoid mistakes whose severity is partly determined by firms' behavior.} We employ this idea in a canonical contracting environment with dynamically inconsistent consumers who may be unaware of their inconsistency (i.e., are naive). Applications are numerous, including credit, subscription, health, retirement, and energy markets.\footnote{More broadly, the framework applies whenever consumers select contracts based partly on behavior they expect from their future selves.} In these environments firms can design contracts to exploit naive consumers. Because avoiding exploitation determines the return to cognitive investment, firms cannot choose how deeply to exploit naive consumers without also affecting how many consumers remain naive. This generates a fundamental trade-off between the intensive margin of exploitation (i.e., its \emph{depth}) and its extensive margin (i.e., its \emph{reach}). More exploitative contracts increase the profit earned from each consumer who remains naive, but also make naivet\'e more costly, inducing greater cognitive investment and thereby reducing the exploitable population. This trade-off disciplines firms and fundamentally changes the evaluation of consumer-protection policies. Interventions that benefit consumers when cognitive states are fixed can weaken the discipline created by cognitive investment, inducing deeper exploitation and potentially reducing consumer welfare.\\

To illustrate this trade-off, consider a consumer choosing how to finance a purchase. One contract offers fixed repayments, while another offers low minimum repayments and the flexibility to repay early. The flexible contract is attractive if the consumer repays early, but costly if they later make only the minimum repayments. When evaluating these contracts, the consumer must form an expectation about how they will repay in the future. A naive consumer expects to repay early, chooses the flexible contract, but subsequently makes only the minimum repayments and accumulates more interest than anticipated. In contrast, a sophisticated consumer anticipates this temptation and selects the commitment provided by the fixed contract. The benefit of sophistication therefore depends directly on the terms of the flexible contract. A lower minimum repayment deepens the exploitation of consumers who remain naive, but also makes naivet\'e more costly, inducing greater cognitive investment and reducing the reach of exploitation.\\

Section~\ref{sec:model} introduces the formal model, which builds on the framework developed by \cite{EliazSpiegler2006}. A profit-maximizing monopolist designs contracts for a continuum of consumers. Contracts provide consumers with the right to select from outcomes, each of which is a consumption-price pair. Consumers are dynamically inconsistent: they select contracts with long-run objectives in mind but give in to temptation when selecting outcomes from their chosen contract. When evaluating contracts, consumers are in one of two cognitive states: sophisticated, in which case they are fully aware of their dynamic inconsistency, or naive, in which case they are unaware of this inconsistency and mistakenly believe that outcomes will be chosen according to their long-run objectives. This misperception leaves naive consumers at risk of exploitation. The key innovation of our paper is that consumers' cognitive states are determined endogenously. Specifically, after observing a menu of contracts, each consumer makes a cognitive-investment decision: consumers invest in cognition and become sophisticated if the depth of exploitation---the long-run welfare loss they avoid by becoming sophisticated---exceeds their idiosyncratic cognitive cost. Otherwise, they do not invest in cognition and remain naive. Consequently, the distribution of cognitive types faced by the monopolist responds to the depth of exploitation implied by the menu of contracts offered, which the monopolist takes into account when designing contracts.\\

There are several distinct advantages of the general contracting environment we employ. First, we do not take a stand on the precise form of consumers' long-run utility or temptation utility. This makes the framework applicable across a wide range of settings, including markets for investment goods, where consumers are tempted to under-consume or under-invest (such as gym memberships or health plans) and leisure goods, where consumers are tempted to over-consume (such as gambling or other vices). Second, we do not impose any \emph{a priori} restrictions on the firm's contract space, and thus abstract from any institutional frictions.\footnote{This is important because restrictions on second- or third-degree price discrimination can generically raise or lower welfare, depending on their effects on rent extraction, consumer sorting, and the volume of trade.} The only non-contractible characteristic is consumers' idiosyncratic cost of becoming sophisticated. The rich contract space therefore isolates the equilibrium consequences of endogenous cognition from those of institutional restrictions on contracting.\footnote{Indeed, the rich contract space makes consumers' private information about their realized cognitive state inconsequential for the results: equilibrium outcomes are as if firms observed whether consumers are naive or sophisticated. Firms nevertheless cannot contract on consumers' cognitive costs, which determine their earlier cognitive-investment decisions.}\\

It should be noted that our formulation of cognitive investment does give rise to an apparent tension: consumers may fail to anticipate how they will behave under a contract, yet evaluate cognition according to the equilibrium consequences of that failure. This feature, however, is common to models of endogenous cognition, which often use optimization not as a literal description of reasoning but as a tractable representation of how cognitive resources respond to incentives.\footnote{Examples include models of rational inattention \citep{Sims2003,MackowiakMatejkaWiederholt2023} when interpreted as models of costly information processing, optimal expectations \citep{Brunnermeier2005}, and motivated beliefs \citep{BenabouTirole2002}. Similarly,  \citet{GulPesendorferStrzalecki2017} model cognitively constrained agents as optimally allocating attention across contingencies and interpret the solution as the outcome of an adjustment process rather than a literal account of reasoning. Empirically, \citet{AlaoiPenta2016} and \citet{gaglianone2022} provide empirical evidence that information-processing and depth of reasoning respond to changes in its costs and benefits.} We adopt such ``rational expectations" because it ties cognitive investment to the same equilibrium outcomes that cognition affects and thereby limits our degrees of freedom. Moreover, this approach still admits natural behavioral interpretations. Under our preferred interpretation, consumers are naive by default but \emph{meta-aware} that deliberate contract evaluation may reveal a failure to anticipate their future behavior.\footnote{Consumers need not enter contract evaluation with a settled view of how they will behave. In the financing example, deliberate evaluation leads the consumer to anticipate making only the minimum repayments, whereas the default process produces the expectation of repaying early. Cognitive investment determines which evaluation process forms the consumer's expectation. In this interpretation, cognitive costs may represent mental effort, deliberation time, information gathering, or the monetary and non-monetary costs of professional advice. More generally, \citet{LoewensteinODonoghueRabin2003} distinguish awareness of a particular bias from meta-awareness of one's general propensity to exhibit it.} Alternatively, our static framework may summarize a dynamic process through which experience, third-party advice, or evolutionary adjustment makes consumers more likely to become sophisticated when the consequences of naivet\'e are more severe.\footnote{\citet{Ali2011} provides a dynamic learning foundation broadly consistent with this interpretation, while \citet{LeYaouanqSchwardmann2022} provide experimental evidence that experience reduces naivet\'e about future behavior. Aspiration-based learning similarly makes behavioral revision more likely following larger payoff shortfalls \citep{KarandikarEtAl1998}.}\\

Section~\ref{sec:results} characterizes the monopolist's optimal contracts under exogenous and endogenous cognition and compares the two benchmarks to isolate the economic forces introduced by cognitive investment. Under exogenous cognition, the rich contract space we allow permits the monopolist to implement its preferred outcome for each cognitive state: the long-run efficient allocation for sophisticated consumers and the profit-maximizing exploitative allocation for naive consumers. The distribution of cognitive types therefore scales profits without affecting the optimal depth of exploitation. Endogenous cognition weakly reduces this depth: at the exogenous optimum, the direct marginal return to deeper exploitation is zero, whereas the induced reduction in reach is costly. Cognitive investment thus deters exploitation relative to any exogenous-cognition benchmark, protecting even consumers who ultimately remain naive.\\

Nonetheless, endogenous cognition does not eliminate exploitation in any environment in which exploitation would persist under exogenous cognition. Thus, there remains a role for policy.  Section~\ref{sec:PA} examines three widely studied interventions intended to improve consumer outcomes: increased competition, sophistication-promoting interventions, and preference nudges. When cognitive types are exogenously fixed, all three improve consumer outcomes. Competition transfers rents from the monopolist to consumers without affecting total welfare, decreasing the depth of exploitation. Sophistication-promoting interventions increase the fraction of sophisticated consumers and thereby reduce the reach of exploitation. Preference nudges make some consumers dynamically consistent, thereby disciplining the contracts offered by the monopolist and reducing both the depth and reach of exploitation. In contrast, with endogenous cognition, the depth--reach tradeoff becomes central to evaluating each policy.\\


Under competition (Section~\ref{subsec:competition}), firms compete by offering contracts that are maximally attractive to each cognitive type. The resulting competitive contracts determine exploitation depth as a by-product: any reduction in exploitation consistent with breaking even would make the contract less attractive to naive consumers, causing the firm to lose them to a rival. As such, competing firms do not internalize how greater exploitation depth reduces its reach, and competition can increase the depth of exploitation relative to monopoly. While competition also benefits consumers by transferring rents from firms, we show that the additional cognitive investment induced by deeper exploitation can be sufficiently costly to reduce consumer welfare overall. Interestingly, whenever competition reduces consumer welfare, it necessarily increases the depth of exploitation and induces greater investment in sophistication, thereby reducing its reach. This decline in reach could be mistaken for market-driven de-biasing, even though it instead reflects costly cognitive defense against deeper exploitation.\\

Under endogenous cognition, there are two forms of sophistication-promoting intervention: those that reduce cognitive costs (i.e., ex-ante policies) and those that improve the mapping from cognitive investment into cognitive states (i.e., ex-post policies). Ex-ante sophistication-promoting interventions (Section~\ref{ssec:EA}) must be evaluated not only by their direct effect on the reach of exploitation, but also by how they change the responsiveness of that reach to exploitation depth. An intervention concentrated among consumers with relatively low cognitive costs removes from the naive population those most likely to become sophisticated as exploitation deepens. The remaining naive population is consequently less responsive to exploitation, weakening the discipline imposed on the monopolist. The resulting increase in equilibrium exploitation depth can be sufficiently large to reduce consumer welfare. Moreover, observed sophistication need not provide a reliable measure of policy effectiveness: an effective, welfare-improving intervention can reduce exploitation so substantially that fewer consumers find cognitive investment worthwhile, causing the equilibrium fraction of sophisticated consumers to fall.\\

Ex-post sophistication-promoting interventions (Section~\ref{ssec:EP}) instead improve consumers' prospects of becoming sophisticated after their initial cognitive-investment decisions. Anticipating this later protection reduces the return to investing ex ante, creating \emph{cognitive moral hazard}. The monopolist often responds by increasing the depth of exploitation, making naive consumers worse off. Indeed, this response can be sufficiently strong to outweigh the policy's direct sophistication benefits, so that strengthening the intervention reduces consumer welfare.\\
 
Finally, \emph{preference nudges} (Section~\ref{subsec:PN}) make a fraction of consumers dynamically consistent and generate two opposing effects on the depth of exploitation. First, because dynamically consistent and naive consumers are indistinguishable when selecting contracts, distortions embedded in contracts designed for naive consumers also reduce the welfare the monopolist can extract from dynamically consistent consumers, raising the marginal cost of deeper exploitation. Second, as more consumers become dynamically consistent, the contractual distortions required to sustain any given depth of exploitation increase, reducing the monopolist's excess profit from a consumer remaining naive rather than becoming sophisticated, reducing incentives to keep consumers naive. When this latter effect dominates, a stronger nudge can increase exploitation depth and reduce consumer welfare.\\

Several of these policy effects have precedents in the literature. Existing work shows that competition, sophistication-promoting interventions, and preference nudges can generate unintended equilibrium consequences, but these results often arise through intervention-specific mechanisms.\footnote{Under competition, exploitation may persist because firms cannot profitably debias consumers \citep{GabaixLaibson2006}, intensify through obfuscation \citep{Spiegler2006,Carlin2009}, or increase through firms' provision of inferior products \citep{GampKrahmer2022}. Cognition-promoting interventions may weaken price competition \citep{ArmstrongChen2009,deClippelEtAl2014}, erode cross-subsidies sustained by naive consumers \citep{KosfeldSchuwer2017}, or induce firms to renew obfuscation \citep{CarlinManso2011}. Other consumer-protection policies may crowd out private information acquisition \citep{ArmstrongVickersZhou2009}, raise equilibrium prices or alter product design in response to disclosure \citep{KamenicaMullainathanThaler2011,Spiegler2015}, or make the exploitation of naive consumers relatively more attractive \citep{MurookaSchwarz2019}. See also \citet{Grubb2015} on regulation that substitutes for private attention and \citet{PiccioneSpiegler2012} on firms' responses to improvements in comparability.} We instead obtain these possibilities within a single contracting framework by introducing one common and economically natural margin: consumers’ incentives to become sophisticated respond to the exploitation they face.\footnote{More broadly, our use of a common mechanism across market and policy environments relates to \citet{EliazSpiegler2015}, who advocate moving behavioral industrial organization beyond isolated bias--market combinations.} Policy changes the depth of exploitation, depth changes consumers' incentives to become sophisticated, and the resulting change in reach feeds back into firms' contracting decisions. The depth--reach mechanism therefore isolates a general source of unintended policy effects: an intervention that reduces one margin of exploitation may weaken the discipline it imposes on the other, potentially reversing its effect on consumer welfare.

\section{Related Literature} \label{sec:literature}  %
 Our paper contributes to two strands of literature: behavioral industrial organization and endogenous cognition in markets. The behavioral-industrial-organization literature primarily studies consumer exploitation while taking the cognitive composition of the market as given. In the context of dynamic inconsistency, \citet{dvm2004} show how firms design contracts around consumers' mistaken predictions of their future behavior, while \citet{EliazSpiegler2006} develop a general contracting framework in which dynamically inconsistent consumers differ in their degree of naivety and a monopolist screens them through commitment and exploitative contracts. More broadly, behavioral-industrial-organization models study how firms respond to consumers who overlook contingent charges \citep{WoodwardHall2012}, mispredict future demand \citep{GrubbOsborne2015}, search too little or remain excessively inert \citep{Grubb2015BestPrice,Armstrong2015}, and overweight salient attributes \citep{BordaloGennaioliShleifer2016}; see \citet{Grubb2015Overview} and \citet{HeidhuesKoszegi2018BIO} for overviews. These models richly characterize the contracts, distortions, and rents associated with consumer mistakes, but generally take the prevalence or intensity of the relevant behavioral friction as given. Firms' choices therefore determine the depth of exploitation, but not its reach through the endogenous cognitive response that we study.\\

There are a number of papers that study endogenous cognition in markets. Most examine consumers' optimal acquisition, processing, or interpretation of information about external payoff-relevant objects. These include information about prices or states of the economy \citep{GulPesendorferStrzalecki2017,AngeletosSastry2025}, economic fundamentals or aggregate actions \citep{ColomboFemminisPavan2014,HebertLao2023}, and the perception of private information \citep{Young2022,Mensch2022,BalzerYoung2025}. Related work studies how firms respond to consumers' endogenous information processing \citep{CusumanoFabbriPieroth2024}, while investors in \citet{CarlinManso2011} learn about financial-product offerings that firms may strategically obfuscate. Other applications allow market incentives to shape strategic sophistication regarding the informational content of prices \citep{Martin2017} or beliefs about lending outcomes \citep{BridetSchwardmann2024}. Our model instead concerns cognition about an \emph{internal} object: consumers' awareness of their own dynamic inconsistency. To the best of our knowledge, we are the first to embed costly investment in such awareness into a market-contracting problem in which firms anticipate how their contracts affect consumers' cognitive states. This allows us to identify the resulting trade-off between the depth and reach of exploitation, place it at the center of the firm's contracting problem, and derive its implications for behavioral-policy evaluation.

\section{The Model}\label{sec:model}

\textbf{Model Overview:} Consider an interaction between a profit-maximizing monopolist and a unit mass of dynamically inconsistent consumers (see Figure~\ref{fig:timeline}). The monopolist first offers a menu of contracts. After observing the menu, and before selecting a contract, consumers decide whether to incur an idiosyncratic cognitive cost and become sophisticated. Consumers first select contracts according to their perceived long-run payoffs and then subsequently select outcomes according to their short-run preferences. Sophisticated consumers correctly anticipate this behavior, whereas naive consumers believe that their future choices will reflect their long-run objectives. The return to cognitive investment is therefore the long-run utility gain from selecting a contract as a sophisticated rather than as a naive consumer. Because this return depends on the offered menu, the monopolist influences the equilibrium distribution of cognitive states and anticipates this response when designing contracts.\\

\begin{figure}[!h]
    \centering
\resizebox{\textwidth}{!}{
   \begin{tikzpicture}[node distance=2.5cm]

\def\xL{2}
\def\xM{7}
\def\xR{12}

\node[draw, circle, minimum size=3pt, fill] (start) at (-2,0) {};

\node[draw, minimum width=2.5cm, minimum height=1.6cm] (boxL1) at (\xL,1.8) {
\begin{tabular}{c}
\begin{tabular}{c}
\small{\textbf{Sophistication:}}\\
\small{Contracts evaluated}\\
\small{\textbf{aware} of}\\
\small{time-inconsistency}
\end{tabular}
\end{tabular}
};

\node[draw, minimum width=2.5cm, minimum height=1.6cm] (boxL2) at (\xL,-1.8) {
\begin{tabular}{c}
\small{\textbf{Naivet{\'e}:}}\\
\small{Contracts evaluated}\\
\small{\textbf{unaware} of}\\
\small{time-inconsistency}
\end{tabular}
};

\node[draw, minimum width=2.5cm, minimum height=1.6cm] (box1) at (\xM,1.8) {
\begin{tabular}{c}
\small{Contract $\sigma$ selected}\\
\small{under $U$ preference}\\
\small{anticipating $V$-choice} 
\end{tabular}
};

\node[draw, minimum width=2.5cm, minimum height=1.6cm] (box2) at (\xM,-1.8) {
\begin{tabular}{c}
\small{Contract $\sigma$ selected}\\
\small{under $U$ preferences,}\\
\small{anticipating $U$-choice} 
\end{tabular}
};

\node[draw, minimum width=2.5cm, minimum height=1.6cm] (box3) at (\xR,1.8) {
\begin{tabular}{c}
\small{Option $y \in \sigma$ selected}\\
\small{under $V$}\\
\small{preference}
\end{tabular}
};

\node[draw, minimum width=2.5cm, minimum height=1.6cm] (box4) at (\xR,-1.8) {
\begin{tabular}{c}
\small{Option $y \in \sigma$ selected}\\
\small{under $V$}\\
\small{preference}
\end{tabular}
};

\draw[thick,->] (start) -- (-0.2, 2);
\draw[thick,->] (start) -- (-0.2,-2);

\draw[thick,->] (boxL1) -- (box1);
\draw[thick,->] (boxL2) -- (box2);

\draw[thick,->] (box1) -- (box3);
\draw[thick,->] (box2) -- (box4);

 \draw[dashed, thick] (4.5,-5) -- (4.5,4);
 \draw[dashed, thick] (9.4,-5) -- (9.4,4);
 \draw[dashed, thick] (-0.5,-5) -- (-0.5,4);
 \draw[dashed, thick] (-3.6, -5) -- (-3.6, 4);

\node at (-2,-4) {Cognitive};
\node at (-2,-4.5) {Investment};
\node at (2,-4.5) {Evaluation Phase};
 
\node at (7,-4.5) {Contract Selection};

\node at (12,-4.5) {Outcome Selection};

\node at (-1.5,1.7) {Invest};
\node at (-2,-1.5) {Don't Invest};

\node[draw, minimum width=3.2cm, minimum height=1.6cm] (boxLL) at (-5.5,0) {
\begin{tabular}{c}
\small{Monopolist}\\
\small{designs menu}\\
\small{of contracts}
\end{tabular}
};
\draw[thick,->] (boxLL) -- (start);
\end{tikzpicture}}
    \caption{Model Overview}
    \label{fig:timeline}
\end{figure}
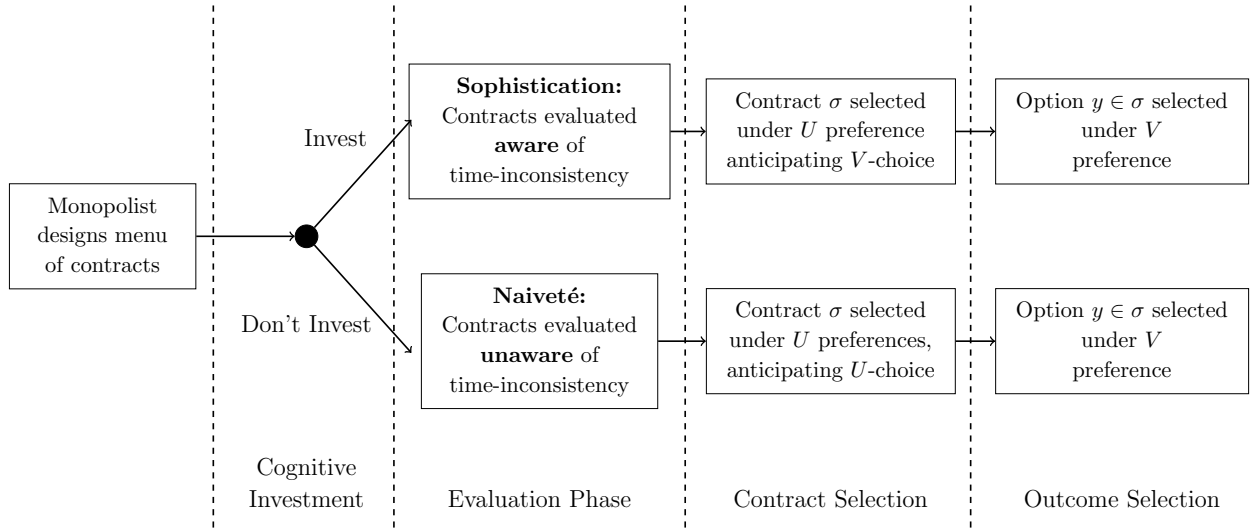

\textbf{Outcomes, Contracts, and the Monopolist:} Consumption is denoted by \(z\in Z=[0,1]\). An outcome is a pair \(y=(\alpha,p)\), where \(\alpha\in\Delta(Z)\) is a distribution over consumption and \(p\in\mathbb{R}\) is a price. With a slight abuse of notation, we write \((z,p)\) for an outcome that delivers consumption \(z\) with probability one. A contract \(\sigma\) is a collection of outcomes, and a menu \(\Sigma\) is a collection of contracts. Consumers first select a contract from the menu and subsequently select an outcome from that contract.\\

Providing consumption \(z\) costs the monopolist \(c(z)\), where
\(c:Z\rightarrow\mathbb{R}\) is continuous and \(c(0)=0\). The profit generated
by outcome \(y=(\alpha,p)\) is
\(
    \pi(y)
    \equiv
    p-\int_0^1 c(z)\,d\alpha(z).
\)
The monopolist designs the menu of contracts \(\Sigma\) to maximize expected profit, taking into account how the menu affects consumers' cognitive-investment, contract-selection, and outcome-choice decisions.\\

\textbf{Consumer Preferences and Dynamic Inconsistency:}
Consumers' preferences over outcomes differ between the contract-selection and
outcome-selection phases. When selecting a contract, consumers evaluate an
outcome \(y=(\alpha,p)\) according to long-run utility
\(
    U(y)
    \equiv
    \int_0^1 u(z)\,d\alpha(z)-p,
\)
whereas, when selecting an outcome from that contract, they maximize short-run
utility
\(
    V(y)
    \equiv
    \int_0^1 v(z)\,d\alpha(z)-p.
\)
The functions \(u\) and \(v\) are continuous and satisfy \(u(0)=v(0)=0\). Because \(u\) and \(v\) may differ, long-run objectives and short-run temptations may favor different outcomes.\\

Every menu contains the outside-option contract \(\bar{\sigma}=\{(0,0)\}\), which gives consumers zero utility under both \(U\) and \(V\). Because \(c(0)=u(0)=v(0)=0\), the maximal values of long-run allocative welfare, \(u(z)-c(z)\), and short-run allocative welfare, \(v(z)-c(z)\), are nonnegative. We assume that each expression is uniquely maximized and denote the respective maximizers by \( z_u^*\equiv\arg\max_{z\in Z}\{u(z)-c(z)\} \) and \( z_v^*\equiv\arg\max_{z\in Z}\{v(z)-c(z)\}. \)\footnote{Uniqueness limits the number of equilibrium outcomes and avoids case distinctions that add no substantive insight.}\\

\textbf{Cognitive States:} When evaluating a menu, a consumer is in one of two cognitive states, \(\omega\in\{S,N\}\): sophisticated (\(S\)) or naive (\(N\)). Sophisticated consumers recognize their dynamic inconsistency and correctly anticipate that their outcome choices will be governed by \(V\). Naive consumers instead believe that they are dynamically consistent and therefore misperceive their outcome choices to be governed by \(U\). Let \(\eta\) denote the fraction of sophisticated consumers, so that \(1-\eta\) is the fraction of naive consumers. A key feature of the model is that \(\eta\) is determined endogenously through consumers' cognitive-investment decisions, as described below. \\

\textbf{Direct Contracts and Directed Menus:} Appendix~\ref{app:cognitive_equilibrium} formally specifies the strategies available to the monopolist and consumers and defines the solution concept we applied, termed \emph{cognitive equilibrium}.\footnote{This term was coined in \citet{Young2022}, where it was applied to misperception of preferences. We adapt this concept to our setting of endogenous awareness of dynamic inconsistency. Specifically, given a menu, consumers are in cognitive equilibrium if (i) they select contracts to maximize perceived long-run utility given their cognitive state, (ii) they subsequently select outcomes according to \(V\), and (iii) their cognitive-investment decisions maximize realized long-run utility net of cognitive costs. Sophisticated consumers correctly anticipate their subsequent outcome choices, whereas naive consumers mistakenly expect these choices to be governed by \(U\).} The appendix also establishes a revelation principle showing that it is without loss of generality to restrict attention to direct contracts and directed menus.\\ 

A contract \(\sigma\) is \emph{direct} if it consists of two outcomes, \(\sigma=\{y_u,y_v\}\), where \(y_u\) is selected under \(U\) and \(y_v\) is selected under \(V\). A direct contract therefore satisfies 

\[ U(y_u)\geq U(y_v) \qquad\text{and}\qquad V(y_v)\geq V(y_u). \] 

Thus, all consumers ultimately select \(y_v\). Sophisticated consumers anticipate this choice, whereas naive consumers mistakenly expect to select \(y_u\).\\ 

A menu is \emph{directed} if it consists of two direct contracts, \(\{\sigma^S,\sigma^N\}\), intended for sophisticated and naive consumers, respectively. Sophisticated consumers evaluate contract \(\sigma^j\) according to \(U(y_v^j)\), correctly anticipating their subsequent \(V\)-choice. Naive consumers instead evaluate it according to \(U(y_u^j)\), believing that their subsequent choice will be governed by \(U\). A directed menu therefore satisfies 

\[ U(y_v^S) \geq \max\bigl\{U(y_v^N),0\bigr\} \qquad\text{and}\qquad U(y_u^N) \geq \max\bigl\{U(y_u^S),0\bigr\}. \] 

The first condition ensures that sophisticated consumers select \(\sigma^S\) and participate, while the second ensures that naive consumers select \(\sigma^N\) and perceive participation as beneficial. We denote the set of directed menus satisfying these constraints by \(\mathcal{M}\).\\

\textbf{Cognitive Investment and Exploitation:}
After observing the menu but before selecting a contract, each consumer decides whether to invest in cognition.\footnote{Naivet\'e about a particular behavioral response does not preclude meta-awareness of one's general susceptibility to such mistakes; see \citet{LoewensteinODonoghueRabin2003}. A consumer may therefore understand the expected value of consulting an adviser, using a decision aid, or engaging in deliberate reflection without already knowing what that process will reveal. Cognitive investment activates such a debiasing technology: before investing, the consumer behaves as if future choices will reflect \(u\); after investing, the consumer correctly anticipates that those choices will instead be governed by \(v\). The benefit of investment is consequently the long-run utility loss that this change in awareness allows the consumer to avoid.} A consumer who invests, \(e=1\), becomes
sophisticated, whereas a consumer who does not invest, \(e=0\), remains naive. Investment entails an idiosyncratic cost \(\kappa\geq0\), distributed across consumers according to a continuous and strictly increasing distribution function \(F\) with support \([0,\bar{\kappa}]\), where \(F(0)=0\) and \[ \bar{\kappa} \geq \max_{z\in Z}\{u(z)-v(z)\} - \min_{z\in Z}\{u(z)-v(z)\}.\footnote{As we will show, the expression on the right is the largest implementable depth of exploitation. Consequently, the assumption ensures that both cognitive states occur with positive probability at every positive interior depth of exploitation.} \] Let $\mathcal{F}$ denote the set of distributions satisfying these properties.\\

The benefit of cognitive investment is equal to the long-run utility gain from selecting a contract as a sophisticated consumer rather than a naive one. A sophisticated consumer selects \(\sigma^S\) and subsequently chooses \(y_v^S\), whereas a naive consumer selects \(\sigma^N\) and
subsequently chooses \(y_v^N\). We therefore define \(x\equiv U(y_v^S)-U(y_v^N).\) A consumer is willing to invest in cognition if \(\kappa \le x\).\footnote{The assumption that consumers correctly evaluate the return to cognitive investment is deliberately conservative. It ties cognition to the utility consequences generated by the offered menu and therefore leaves the modeler no separate freedom to specify how perceived cognitive benefits respond to contracts or policies. Allowing consumers to misperceive these benefits, learn from noisy feedback, rely on advisers of varying quality, or follow other boundedly rational investment rules would introduce additional behavioral degrees of freedom. Such flexibility could generate further divergences between a policy's direct effect and its equilibrium consequences.} Hence, the equilibrium fraction of sophisticated consumers is \( \eta=F(x) \). Through its choice of menu, the monopolist influences the equilibrium distribution of cognitive states. \\

We note that \(x>0\) if and only if a naive consumer makes a suboptimal contract-selection decision. Thus, we refer to \(x\) as the \emph{depth of exploitation}. A directed menu is \emph{exploitative} if \(x>0\), so that naive consumers obtain strictly lower long-run utility than sophisticated consumers. The value \(x\) consequently has a dual interpretation: it measures both the loss from remaining naive and the private return to cognitive investment. Greater depth of exploitation therefore induces more consumers to become sophisticated.\\

\section{Equilibrium Analysis}\label{sec:results}

Our analysis proceeds in two parts. First, in order to benchmark our results, we characterize equilibrium contracts and exploitation when cognitive states are exogenous. Second, we make cognitive states endogenous and characterize the monopolist's optimal menu and the resulting distribution of cognitive states. Comparing the two benchmarks shows how endogenous cognition disciplines exploitation and establishes the mechanisms underlying our subsequent policy analysis.

\subsection{Exogenous Cognitive States}\label{subsec:exo}

We first consider the benchmark in which consumers' cognitive states are exogenous. A fixed fraction \(\bar{\eta}\in[0,1]\) of consumers is
sophisticated, while the remaining fraction \(1-\bar{\eta}\) is naive.
Because consumers' cognitive states do not respond to the available contracts, the monopolist's menu does not affect their distribution. Maximizing profits thus requires solving 
\[
    \max_{\{\sigma^S,\sigma^N\}\in\mathcal{M}}
    \bar{\eta}\,\pi(y_v^S)
    +
    (1-\bar{\eta})\,\pi(y_v^N).
\]
In words, the monopolist chooses a directed menu of contracts to discriminate between sophisticated and naive consumers. The following lemma characterizes the outcomes received by each cognitive type.

\begin{lemma}\label{lemma:contracts_exogenous}
Suppose that cognitive states are exogenous. Then:
\begin{enumerate}[(a)]
    \item sophisticated consumers, if any, consume \(z_u^*\) at price
    \(u(z_u^*)\); and
    \item naive consumers, if any, consume \(z_v^*\) at price 
    \(  
        v(z_v^*)
        +
        \max \limits_{z\in Z}\bigl\{u(z)-v(z)\bigr\} 
    \).
\end{enumerate}
\end{lemma}

The monopolist implements the long-run efficient consumption level \(z_u^*\) for sophisticated consumers and sets the price to extract their entire welfare. Because sophisticated consumers correctly anticipate their subsequent outcome choices, they are not exploited, even though the monopolist leaves them no information rents.\\ 

The contract directed toward naive consumers instead exploits their misperception of their future behavior. When selecting a contract, a naive consumer expects their subsequent outcome choice to maximize long-run utility \(u\). The monopolist attracts them with a ``bait'' outcome that maximizes the difference between its long-run and short-run valuation. The consumer does not ultimately select this outcome. Instead, they yield to temptation and consume \(z_v^*\), as anticipated by the monopolist. The difference between the consumer's anticipated and realized choices allows the monopolist to charge \(v(z_v^*)+\max \limits_{z\in Z}\{u(z)-v(z)\}, \) thereby converting the consumer's misperception into profit.\footnote{This contract structure is consistent with empirical evidence on consumer misprediction: \cite{DellaVigna2006} document that health-club members systematically overestimate future attendance, paying over \$17 per visit when a \$10 per-visit pass is available and forgoing average savings of \$600. The monthly contract exploits the gap between planned and actual usage in precisely the manner described here. \cite{grubb2009} documents the same mechanism in cellular phone markets: overconfident consumers underestimate future usage and choose plans with steep overage charges, which firms design around this predictable mistake.}\\ 

Sophisticated consumers receive zero long-run utility, while naive consumers receive long-run utility \( u(z_v^*) - v(z_v^*) - \max_{z\in Z}\bigl\{u(z)-v(z)\bigr\}  \). 
Thus, define 
\[ \bar{x} \equiv \max_{z\in Z}\bigl\{u(z)-v(z)\bigr\} - \bigl[u(z_v^*)-v(z_v^*)\bigr]. \] 
Then, $\bar{x}$ represents the depth of exploitation in the exogenous benchmark. Thus, the market is exploitative whenever $\bar{\eta} < 1$ and $\bar{x}>0$. The following proposition characterizes when this is the case.

\begin{prop}
\label{prop:exoiff} 
Suppose that cognitive states are exogenous and that the fraction of sophisticated consumers is \(\bar{\eta}\in[0,1]\). The market is exploitative if and only if \( \bar{\eta}<1 \) and \( z_v^* \notin \arg\max_{z\in Z}\bigl\{u(z)-v(z)\bigr\} \). 
\end{prop} 

Proposition~\ref{prop:exoiff} shows that naive consumers are generically exploited. Exploitation is absent only in the special case in which the short-run efficient consumption level \(z_v^*\) also maximizes the preference wedge \(u-v\). In that case, \(z_v^*\) maximizes both \(v-c\) and \(u-v\), and therefore also maximizes their sum \(u-c\). Hence, \(z_v^*=z_u^*\): although \(u\) and \(v\) may differ, they induce the same welfare-maximizing consumption level. Otherwise, \(\bar{x}>0\), and naive consumers receive strictly less long-run utility than sophisticated consumers. \\

Crucially, \(\bar{x}\) is independent of \(\bar{\eta}\). With a rich, non-linear contracting space, the monopolist can separately implement its preferred outcome for each type, so the distribution of cognitive states affects the relative profit earned from the two contracts but not their design. Thus, $\bar{x}$ provides the fixed-cognition benchmark against which we assess the disciplining effect of endogenous cognition.

\subsection{Endogenous Cognition}\label{subsec:endo}

We now endogenize cognition by assuming consumers must invest in order to become sophisticated. A directed menu generating value
of sophistication \(x\) induces a fraction \(F(x)\) of sophisticated consumers
and a fraction \(1-F(x)\) of naive consumers. The monopolist anticipates this
response when designing its menu. We solve its problem in two steps. First, we
characterize the profit-maximizing directed menu that implements a given
\(x\). Second, we determine the monopolist's optimal choice of \(x\).\\

Turning to the first step, for a given \(x\), the monopolist
solves
\begin{equation}\label{eqn:profits}
    \max_{\{\sigma^S,\sigma^N\}\in\mathcal{M}}
    F(x)\pi(y_v^S)
    +
    \bigl[1-F(x)\bigr]\pi(y_v^N)
\end{equation}
subject to
\(
    U(y_v^S)-U(y_v^N)=x 
\).
 That is, the problem is as in the exogenous benchmark but with an additional constraint ensuring that the implemented menu of contracts induces a depth of exploitation equal to $x$. The following lemma provides the endogenous-cognition counterpart to Lemma~\ref{lemma:contracts_exogenous}, characterizing the profit-maximizing menu conditional on implementing a given depth of exploitation.

\begin{lemma}\label{lemma:contracts_endogenous} 
Suppose that cognitive states are endogenous. Conditional on implementing a depth of exploitation \(x\): 
\begin{enumerate}[(a)] 
\item sophisticated consumers receive consumption \(z_u^*\) at price \(u(z_u^*)\); and 
\item naive consumers receive an outcome \(y_v^N=(\alpha_v^N,p_v^N)\), where 
\[ p_v^N = \int_0^1 u(z)\,d\alpha_v^N(z)+x, \] 
and 
\[ \alpha_v^N \in \arg\max_{\alpha\in\Delta(Z)} \int_0^1 \bigl[u(z)-c(z)\bigr]\,d\alpha(z) \] 
subject to \( \max \limits_{z\in Z}\bigl\{u(z)-v(z)\bigr\} - \int_0^1 \bigl[u(z)-v(z)\bigr]\,d\alpha(z) \geq x \). 
\end{enumerate} \end{lemma}

As in the exogenous benchmark, sophisticated consumers receive the long-run efficient consumption level and are left with zero utility. The value-of-sophistication constraint therefore requires naive consumers to receive long-run utility \(U(y_v^N)=-x\). The monopolist implements this utility by setting \(p_v^N=\int_0^1 u(z)\,d\alpha_v^N(z)+x\).\\ 

Relative to the exogenous benchmark, the adjustment therefore occurs in the realized outcome of the naive contract. Conditional on \(x\), the monopolist can extract any increase in the long-run allocative welfare generated by \(y_v^N\) through a higher price, without changing the naive consumer's long-run utility or the depth of exploitation. It therefore chooses the naive consumer's consumption distribution to maximize long-run allocative welfare, subject to implementing \(x\). In this sense, the monopolist internalizes the allocative consequences of the naive consumer's realized outcome even though it does not directly value consumer welfare. Indeed, when the implementability constraint is slack, the naive consumer receives the long-run efficient allocation.\\

When the long-run efficient allocation does not satisfy the implementability constraint, the monopolist must distort consumption, giving rise to an allocative cost of exploitation. To formalize this cost, let \(W^*\equiv\max_{z\in Z}\{u(z)-c(z)\}\) denote maximal long-run allocative welfare, and define 

\[ \begin{aligned} L(x) \equiv \min_{\alpha\in\Delta(Z)} \quad & W^*-\int_0^1 \bigl[u(z)-c(z)\bigr]\,d\alpha(z) \\[-0.2em] &\hspace{-3.7em}\text{s.t.}\quad \max_{z\in Z}\bigl\{u(z)-v(z)\bigr\} -\int_0^1 \bigl[u(z)-v(z)\bigr]\,d\alpha(z) \geq x. \end{aligned} \] 

The constraint requires the gap between the preference wedges generated by the bait and realized outcomes to be at least \(x\). If the long-run efficient allocation does not generate a sufficiently large gap, the monopolist must distort the naive consumer's realized consumption away from long-run efficiency. Thus, \(L(x)\) is the smallest loss of long-run allocative welfare required to implement exploitation \(x\). The sophisticated contract generates profit \(W^*\), while the naive contract generates profit \(W^*+x-L(x)\). Hence, \(x-L(x)\) is the excess profit earned from a naive consumer relative to a sophisticated consumer.\\

We now turn to the monopolist's choice of \(x\). Let \(x_{\max}\equiv\max_{z\in Z}\{u(z)-v(z)\} -\min_{z\in Z}\{u(z)-v(z)\}\). Lemma~\ref{lemma:contracts_endogenous} implies that \(x\) is implementable if and only if \(x\in[0,x_{\max}]\).\footnote{The choice set is therefore compact, so an optimal depth of exploitation exists.} Using the optimal contracts conditional on \(x\) and the fact that \(W^*\) is independent of \(x\), the monopolist's problem is equivalent to \begin{equation}\label{eqn:optimal_x} \max_{x\in[0,x_{\max}]} \bigl[1-F(x)\bigr]\bigl[x-L(x)\bigr]. \end{equation} 
Let \(x^*\) denote a solution to \eqref{eqn:optimal_x}.
\footnote{Because $F$ is continuous and $x-L(x)$ is concave, the objective is upper semicontinuous. Compactness of the choice set therefore implies that the set of maximizers is nonempty and compact.} Thus, despite the model's generality, the monopolist's problem reduces to a one-dimensional choice of exploitation depth.\\

The reduced-form problem reveals the monopolist's trade-off. Conditional on remaining naive, a consumer generates excess profit \(x-L(x)\). Greater exploitation, however, induces more consumers to invest in cognition and reduces the reach of exploitation (i.e., $1 - F(x)$ decreases). Because the monopolist serves the entire market, it internalizes this impact on the reach of exploitation when choosing \(x\). This induces the monopolist to choose weakly less exploitation relative to the exogenous benchmark.

\begin{prop}\label{prop:sanity}
    Suppose that $\bar{x} > 0$. Relative to the exogenous benchmark, endogenous cognition disciplines the monopolist but never eliminates exploitation: every equilibrium depth of exploitation $x^*$ satisfies $0 < x^* \le \bar{x}$.
\end{prop}

Proposition~\ref{prop:sanity} establishes that cognitive investment acts as a deterrence mechanism. Under exogenous cognition, \(\bar{x}\) maximizes the excess profit earned from each naive consumer, independently of the distribution of cognitive types. With endogenous cognition, the monopolist must also account for the loss of reach induced by deeper exploitation. It consequently chooses weakly less exploitation than under the exogenous benchmark, protecting even consumers who ultimately remain naive. \\

Nonetheless, endogenous cognition does not eliminate exploitation. Naive consumers are still more profitable than their sophisticated counterparts and, consequently, the monopolist will always choose a positive depth of exploitation whenever profitable scope for exploitation exists. As such, endogenous cognition disciplines the market without eliminating the potential role for consumer-protection policy.

\section{Policy Analysis}\label{sec:PA}

The benchmark market is generically exploitative because consumers are
dynamically inconsistent, becoming aware of this inconsistency is cognitively
costly, and firms can profit from consumers' misperceptions of their future
behavior. We consider several consumer-protection policies, each of which targets one of these frictions. Competition limits firms' ability to appropriate
the rents generated by consumer mistakes. Sophistication-promoting policies
reduce cognitive frictions or otherwise help consumers recognize their dynamic
inconsistency. Preference nudges instead reduce the underlying conflict between
long-run objectives and short-run temptation. 
The targeted primitive differs across
interventions: ex-ante sophistication-promoting policies change the distribution
of cognitive costs, ex-post sophistication-promoting policies change the
probability that a consumer becomes sophisticated without investing, and
preference nudges change the fraction of consumers whose choices are dynamically
consistent.\\

We call an intervention effective if it successfully changes its intended target in the intended direction. For example, ex-ante sophistication-promoting policies are effective when they reduce cognitive costs, ex-post sophistication-promoting policies are effective when they facilitate sophistication directly, and preference nudges are effective when they reduce the prevalence of dynamic inconsistency. A desirable consumer-protection policy should translate this effectiveness into improved consumer outcomes. Because consumers in our framework endogenously invest in sophistication, evaluating these outcomes must account not only for the consumption allocations they receive but also for the cognitive resources they expend to achieve them. Our primary welfare measure is therefore consumer welfare, inclusive of cognitive costs. That is, our primary welfare measure coincides with ex-ante consumer utility.\\

The distinctive feature of our framework is precisely this endogeneity of sophistication. Depending on its target, a policy may directly affect consumers’ cognitive incentives or states, or it may alter the contracting environment and thereby affect cognition indirectly. In either case, consumers’ cognitive responses change firms’ contract-design incentives, while firms’ contractual responses change consumers' sophistication incentives. As such, policy generates an equilibrium interaction between contract design, exploitation, and cognition that is absent when consumers’ cognitive states are held fixed. Our analysis asks whether a policy that is effective and improves consumer outcomes when cognitive states are exogenous continues to do so once these behavioral and contractual responses are taken into account.\\

We show that each policy is well designed under the conventional
assumption of exogenous cognitive states: whenever the intervention is
effective, it improves consumer welfare. In contrast, this conclusion need not survive when
cognition is endogenous. Consumers may adjust their cognitive investment, and
firms may respond by changing the depth of exploitation. These equilibrium
responses can weaken or even overturn the policy's direct benefits. A complete
welfare evaluation must therefore account for both the policy's intended effect
and the equilibrium responses it induces.\\

Endogenous sophistication also creates difficulties in identifying whether stronger policies are more effective, and hence in performing program evaluation. This issue is particularly acute for ex-ante sophistication-promoting policies in our framework.\footnote{The empirical literature documents substantial heterogeneity in the effectiveness of interventions intended to improve financial decision-making. \cite{kaiser2024}, for example, find that many financial-literacy programs fail to improve financial behaviour even when contracts are held fixed.} Such interventions are intended to reduce the cost of becoming sophisticated. However, observed sophistication reflects not only cognitive costs but also the equilibrium value of self-protection. An effective policy may lower cognitive costs and induce firms to reduce exploitation. The resulting decline in the benefits of sophistication can then reduce equilibrium sophistication, even though the intervention improves consumer welfare. Conversely, an increase in sophistication may reflect greater exploitation rather than a successful reduction in cognitive frictions. Broadly, policy effectiveness, observed behavioral responses, and consumer welfare can come apart. Conditions that make welfare effects transparent need not make effectiveness identifiable from observed sophistication, and vice versa.\\

Throughout the remainder of this section, we restrict attention to environments in which exploitation arises in equilibrium: we assume $z_v^* \notin \arg \max \limits_{z \in Z} [u(z) - v(z)]$, so that $\bar{x} > 0$ (Propositions~\ref{prop:exoiff} and \ref{prop:sanity}). The analysis below therefore focuses on settings in which positive scope for exploitation remains and consumer-protection policy may be warranted. Moreover, whenever multiple equilibria exist, we select the one that minimizes the depth of exploitation.\footnote{This equilibrium is Pareto dominant in all contexts. This convention also ensures that none of our findings are generated by coordination failures associated with particular equilibrium selections.}

\subsection{Competition}\label{subsec:competition}

Historically, competition has been viewed as a natural market-based antidote to consumer exploitation. The conventional argument has two parts. First, competitive pressure transfers rents from firms to consumers, limiting firms' ability to profit from consumer mistakes and increasing consumer welfare. Second, firms may attract consumers by offering less distorted contracts or by exposing exploitative practices, thus de-biasing consumers. Competition should therefore benefit consumers by decreasing both the depth and reach of exploitation.\\

The behavioral literature has already qualified this conventional perspective. The argument that competition reduces the reach of exploitation is weakened by the \emph{curse of de-biasing}: firms may be unable to profit from exposing exploitative practices and correcting consumers' mistakes \citep{GabaixLaibson2006}. Moreover, profitable de-biasing may serve simply to soften competition and extract greater rents from those who remain naive \citep{johnen2020}. Nor must competition reduce the depth of exploitation. Competitive markets can sustain distortions \citep{dvm2004,HeidhuesKoszegi2010} and the provision of low-quality products \citep{ArmstrongChen2009} even when competition dissipates firms' rents. In such settings, competition may still benefit consumers through lower transparent prices, although these benefits may accrue disproportionately to sophisticated consumers \citep{GabaixLaibson2006,armstrong2012}. In other settings, competition can even reduce consumer welfare by increasing firms' use of obfuscation \citep{Spiegler2006,Carlin2009}, encouraging the provision of inferior products \citep{GampKrahmer2022}, or inducing additional biased consumers to make welfare-reducing purchases \citep{as2019}. We complement this literature by showing that competition can have adverse consequences through endogenous cognition and that a resulting reduction in reach need not indicate successful de-biasing but, rather, signal a greater depth of exploitation.\\

To analyze competition, we introduce $N \ge 2$ firms that simultaneously offer directed menus of contracts. After observing all menus, consumers make their cognitive-investment and contract-selection decisions as in the benchmark model.\footnote{Specifically, the return to cognition depends only on the long-run utility difference between the contracts selected under sophistication and naivet{\'e}. Unselected contracts therefore play no role, so our competition results are not driven mechanically by the presence of more contracts to evaluate.} Because firms can compete over a rich space of contracts, any $N \ge 2$ is sufficient to implement the perfectly competitive allocation, as in Bertrand competition.\footnote{Because each firm can offer contracts directed toward both cognitive states, competition cannot be softened through market segmentation. Moreover, continuity of the cognitive-cost distribution ensures that a sufficiently small contractual improvement leaves a positive mass of the targeted cognitive type. Thus, if a contract selected by either type failed to maximize that type's perceived utility subject to nonnegative profit, a rival could offer a slightly more attractive and strictly profitable contract to a positive mass of those consumers.} The following lemma summarizes consumer outcomes under such perfect competition.

\begin{lemma}\label{lemma:contracts_comp} 
Whether cognition is exogenous or endogenous, under competition, sophisticated consumers receive $\bigl(z_u^*,c(z_u^*)\bigr)$, whereas naive consumers receive $\bigl(z_v^*,c(z_v^*)\bigr)$. 
\end{lemma} 

Firms compete for each cognitive type according to how that type evaluates contracts. Sophisticated consumers correctly anticipate their subsequent outcome choices, so competition implements the long-run efficient allocation and leaves them with the entire welfare. Naive consumers instead evaluate contracts according to the outcome they mistakenly expect to choose. Competition therefore maximizes their perceived utility subject to firms breaking even, but their realized consumption remains \(z_v^*\). Thus, competition eliminates firms' profits without necessarily maximizing naive consumers' long-run welfare. The resulting depth of exploitation is 
\[ 
x^C \equiv \max_{z\in Z}\bigl\{u(z)-c(z)\bigr\} - \bigl[u(z_v^*)-c(z_v^*)\bigr], 
\] 
which can remain strictly positive.\\

We first consider the benchmark in which cognitive states are exogenous. With cognitive states held fixed, competition does not change the consumption outcomes of either cognitive type relative to monopoly (see Lemma~\ref{lemma:contracts_exogenous}). Sophisticated consumers continue to consume $z_u^*$, while naive consumers continue to consume $z_v^*$. Competition therefore does not directly correct the allocative distortion generated by naive consumers' misperceptions. Instead, it changes how the welfare generated by these consumption outcomes is divided. Under competition, firms break even and consumers receive the entire welfare from their respective outcomes. Under monopoly, the firm extracts this welfare and additionally benefits from naive consumers' misperceptions. Competition therefore weakly reduces the depth of exploitation and increases consumer welfare, despite leaving the underlying allocative distortion unchanged.

\begin{prop}\label{prop:comexo} Suppose that cognitive states are exogenous, with the fraction of sophisticated consumers fixed at $\bar{\eta}\in[0,1]$. Then, $x^C\leq\bar{x}$, and consumer welfare is higher under competition than under monopoly. 
\end{prop}
 
Thus, when cognition is fixed, competition dissipates firms' rents, reduces
exploitation, and increases consumer welfare. It need not, however, eliminate the underlying behavioral distortion or give firms a positive incentive to de-bias consumers, consistent with prior findings under exogenous cognitive types \citep{GabaixLaibson2006}.\footnote{Indeed, although firms cannot directly transform naive consumers into sophisticated ones in our model, such a transformation would be payoff-irrelevant in equilibrium because competitive firms earn zero profits from both types.}\\

With endogenous cognition, the welfare comparison depends on whether competition raises or lowers exploitation relative to monopoly. Let \(x^*\) denote the equilibrium depth of exploitation under monopoly. The fractions of sophisticated consumers under competition and monopoly are then \(F(x^C)\) and \(F(x^*)\), respectively.\\ 

Two opposing forces determine the comparison between \(x^C\) and \(x^*\). On the one hand, competition transfers rents from firms to consumers, which tends to reduce exploitation. On the other hand, Bertrand competition forces firms to maximize the perceived attractiveness of their contracts. The resulting competitive contract determines the depth \(x^C\) as a by-product: a firm cannot reduce \(x^C\) while preserving the contract's perceived attractiveness and would therefore lose its naive customers to a rival if it offered a less exploitative contract. As such, unlike the monopolist, competitive firms do not choose exploitation to balance depth against reach. Consequently, \(x^C\) is independent of the cognitive-cost distribution \(F\), and competition can increase the depth of exploitation relative to monopoly and may even decrease consumer welfare. The following result characterizes precisely when competition unambiguously improves consumer outcomes.

\begin{prop}\label{prop:comendo}
Suppose that cognitive states are endogenous. Then, consumer welfare is higher
under competition than under monopoly for every cognitive-cost distribution
\(F \in \mathcal{F}\) if and only if, for every \(x\in[0,x^C]\),
\[
    \int_x^{x^C}
    \frac{x-L(x)}{s-L(s)}
    \,ds
    \leq
    \max_{z\in Z}\bigl\{u(z)-c(z)\bigr\}.
\]
\end{prop}

To understand the condition, suppose that the monopolist chooses exploitation \(x\). For this choice to be optimal relative to any higher level \(s>x\), it must satisfy
\(
    \bigl[1-F(s)\bigr]\bigl[s-L(s)\bigr]
    \leq
    \bigl[1-F(x)\bigr]\bigl[x-L(x)\bigr]
\)
or, equivalently,
\(
    1-F(s)
    \leq
    \bigl[1-F(x)\bigr]
    \frac{x-L(x)}{s-L(s)}.
\)
Conditional on \(x\) being the monopoly optimum, this inequality bounds the mass of consumers who can remain naive at any higher level of exploitation \(s\). If the naive mass were larger, the monopolist would prefer \(s\), contradicting the optimality of \(x\). Hence, a cognitive-cost distribution that attains this bound at every \(s>x\) leaves as many consumers naive as is consistent with \(x\) being the monopoly optimum. Consequently, it generates the largest additional cognitive costs that can arise when competition raises exploitation from \(x\) to \(x^C\). The integral in Proposition~\ref{prop:comendo} aggregates these worst-case costs. The right-hand side is the long-run allocative welfare passed to a sophisticated consumer under competition. The condition requires these contractual gains from competition to dominate the additional cognitive costs even under this worst-case distribution. It is therefore both necessary and sufficient for competition to increase consumer welfare independently of \(F\). \\

The condition is more likely to be satisfied when competition generates greater long-run allocative welfare and when the outcome received by naive consumers is less exploitative. This suggests that competition is more likely to be unambiguously beneficial in markets for \emph{investment goods}, such as fitness or health, for which the underlying behavioral problem is underconsumption, than for \emph{temptation goods}, such as gambling or smoking, for which the underlying problem is overconsumption.\\

When this condition fails, there exist cognitive-cost distributions for which consumer welfare is lower under competition than under monopoly. Interestingly, whenever competition reduces consumer welfare, it necessarily holds that $x^C > x^*$ and, consequently, $F(x^C) \ge F(x^*)$. That is, the reach of exploitation necessarily \emph{decreases} whenever competition policy backfires. This creates a delicate problem for evaluating competition policies: upon observing that competition decreases the fraction of consumers exploited, conventional de-biasing arguments would suggest that consumers have become less susceptible to exploitation. In contrast, interpreted through the lens of our framework, this is strongly indicative of an increase in the depth of exploitation and a potential worsening of consumer outcomes.

\subsection{Sophistication-Promoting Policies}\label{subsec:exante}

Sophisticated consumers correctly anticipate their subsequent outcome choices and are therefore not exploited. This makes policies that promote sophistication a natural response to consumer naivet\'e. When cognitive states are exogenous, such policies simply reduce the fraction of consumers exposed to exploitation. With endogenous cognition, they also change consumers' incentives to invest in cognition and the monopolist's incentives to design exploitative contracts. We index sophistication-promoting interventions by a policy parameter \(\tau\in T\subseteq[0,1]\), where higher values correspond to stronger interventions.\footnote{Here, $\tau$ serves as a placeholder for a policy intervention; its precise interpretation is specified in the relevant section.}\\

We first consider the benchmark in which cognitive states are exogenous. In this case, it is natural to model a sophistication-promoting policy as operating directly on the fraction of sophisticates. Thus, let $\tau \in [0,1]$ denote a policy parameter and $\bar{\eta}(\tau)$ denote the fraction of sophisticates, where $\bar{\eta}(\cdot)$ strictly increases, so that a higher $\tau$ represents a more effective intervention. Recall that, by Lemma~\ref{lemma:contracts_exogenous}, the contracts received by sophisticated and naive consumers do not vary with their relative prevalence. Consequently, increasing $\tau$ reduces the reach of exploitation without changing its depth. Promoting sophistication therefore unambiguously benefits consumers in the exogenous benchmark.

\begin{prop}\label{prop:exoeasp} 
Suppose that consumers' cognitive states are exogenous. Then, consumer welfare increases as sophistication-promoting policies become more effective (i.e., as $\tau$ increases). 
\end{prop}

With endogenous cognition, sophistication-promoting policies can operate either before or after consumers make their cognitive-investment decisions. As such, we distinguish between \emph{ex-ante policies}, which reduce the cost of cognitive investment, and \emph{ex-post policies}, which improve the mapping from cognitive-investment decisions into cognitive states. Figure~\ref{fig:pstvsante} illustrates this distinction.

\begin{figure}
    \centering
    \begin{tikzpicture}[
    >=stealth,
    every node/.style={align=center}
]

\begin{scope}[xshift=0cm]

\node at (3,2.3) {\small Ex-ante policy};

\node (kappa) at (0,0) {$\kappa$};
\node (kappatau) at (2,0) {$\kappa_{\tau}$};

\draw[->, thick]
    (kappa) -- node[above] {\small policy} (kappatau);

\node (e1) at (4,1) {$e=1$};
\node (e0) at (4,-1) {$e=0$};

\draw[->, thick]
    (kappatau) -- node[above, sloped] {\small pay $\kappa_{\tau}$} (e1);

\draw[->, thick]
    (kappatau) -- node[below, sloped] {\small pay $0$} (e0);

\node (S) at (6,1) {$S$};
\node (N) at (6,-1) {$N$};

\draw[->, thick] (e1) -- (S);
\draw[->, thick] (e0) -- (N);

\end{scope}

\begin{scope}[xshift=8.5cm]

\node at (3,2.3) {\small Ex-post policy};

\node (kappar) at (0,0) {$\kappa$};

\node (e1r) at (2,1) {$e=1$};
\node (e0r) at (2,-1) {$e=0$};

\draw[->, thick]
    (kappar) -- node[above, sloped] {\small pay $\kappa$} (e1r);

\draw[->, thick]
    (kappar) -- node[below, sloped] {\small pay $0$} (e0r);

\node (Sr) at (6,1) {$S$};
\node (Nr) at (6,-1) {$N$};

\draw[->, thick] (e1r) -- (Sr);

\draw[->, thick]
    (e0r) -- node[above, sloped] {$\tau$} (Sr);

\draw[->, thick]
    (e0r) -- node[below, sloped] {$1-\tau$} (Nr);

\end{scope}

\end{tikzpicture}

    \caption{Ex-ante and ex-post sophistication-promoting policies. Ex-ante
    policies lower $\kappa_{\tau}$, the cost of cognitive investment. Ex-post policies instead
    increase the probability that a consumer becomes sophisticated after the
    investment decision.}
    \label{fig:pstvsante}
\end{figure}
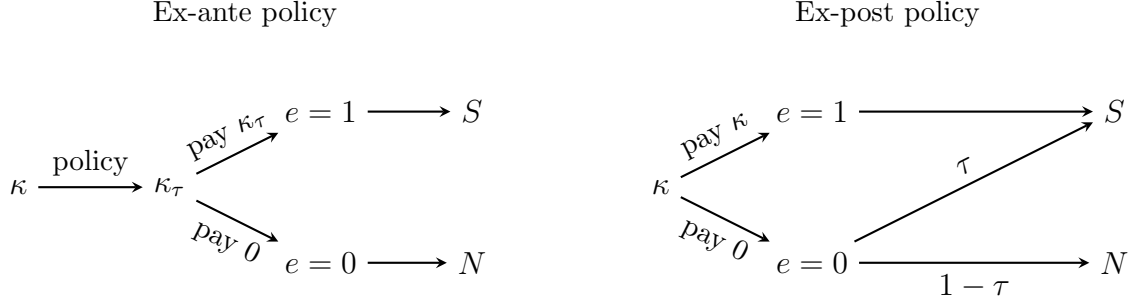

\subsubsection{Ex-Ante Sophistication-Promoting Policies}\label{ssec:EA}

Ex-ante policies reduce cognitive frictions at the contract-evaluation stage. Examples include financial education, decision aids, and access to professional or technological advice. We model these interventions as reductions in consumers' cognitive costs. Let $\bigl\{F(\cdot;\tau)\bigr\}_{\tau\in T}$ denote a family of policy-induced cognitive-cost distributions, where $\tau$ indexes policy intensity. Because ex-ante policies are intended to facilitate cognitive investment, we call such an intervention \emph{effective} if stronger policies induce first-order stochastically lower cognitive costs: $\tau'\geq\tau$ implies $F(\kappa;\tau')\geq F(\kappa;\tau)$ for every $\kappa$. It follows that, holding contracts (and the corresponding depth of exploitation) fixed, an effective increase in policy intensity weakly increases sophistication. Correspondingly, the monopolist chooses the optimal depth of exploitation to solve the problem in \eqref{eqn:optimal_x}, with $F(x)$ replaced by $F(x; \tau)$. \\

In the exogenous benchmark, a policy that increases the fraction of sophisticated consumers necessarily improves consumer outcomes because the depth of exploitation does not respond to the policy. The analogous conclusion does not follow directly from an improved distribution of cognitive costs when cognition is endogenous. What matters is not only \emph{whether} cognitive costs fall, but also \emph{how} the policy reshapes the distribution of cognitive costs across consumers. The following condition defines what we call a progressive change in this composition. 

\begin{definition}
A policy intervention is \emph{progressive} if, for every \(\kappa > \kappa^\prime\) such that
\(1-F(\kappa^\prime;\tau)>0\),
\[
    \frac{1-F(\kappa;\tau)}{1-F(\kappa^\prime;\tau)}
\]
is non-increasing in \(\tau\).
\end{definition}

A progressive policy thins the upper tail of the cognitive-cost distribution. For any two thresholds $\kappa>\kappa^\prime$, the share of consumers with costs above $\kappa$, relative to the share with costs above $\kappa^\prime$, falls as the policy strengthens. In turn, the population remaining naive becomes more responsive to exploitation: a further increase in \(x\) produces a larger proportional reduction in the fraction of naive consumers. This strengthens the discipline imposed by cognitive investment, as deeper exploitation more rapidly erodes the monopolist’s relatively more profitable naive-consumer base. The monopolist consequently chooses a lower equilibrium depth of exploitation, reinforcing the direct benefits of lower cognitive costs. The following proposition formalizes these effects.

\begin{prop}\label{prop:enoeasp}
Suppose that consumers' cognitive states are endogenous and that
sophistication-promoting policies are both effective and progressive. Then, as policy intensifies (i.e., $\tau$ increases), the equilibrium
depth of exploitation, \(x^*(\tau)\), decreases and consumer
welfare increases.
\end{prop}

Proposition~\ref{prop:enoeasp} establishes that progressivity is sufficient for an effective policy to improve consumer welfare. Not all effective policies are progressive, however, and consequently are not guaranteed to improve consumer outcomes. Consider an effective policy that disproportionately lowers the costs of consumers who initially have relatively low cognitive costs. In the absence of the policy, these consumers discipline the monopolist because their cognitive-investment decisions are particularly responsive to exploitation. As their costs fall, however, they become sophisticated even at relatively low levels of exploitation. The residual pool on the margin of remaining naive consequently becomes more concentrated among high-cost types and, thus, is less responsive to further increases in exploitation. This weakens the discipline imposed by cognitive investment and may induce the monopolist to increase the equilibrium depth of exploitation. The resulting increase in exploitation can outweigh the direct benefits of lower cognitive costs, causing consumer welfare to fall as the policy strengthens.\footnote{\citet{KosfeldSchuwer2017} likewise show that effective consumer education can reduce welfare after firms adjust. In their model, educating some consumers reduces firms' revenue from hidden charges, inducing higher prices for other consumers and costly avoidance of those charges. \citet{CarlinManso2011} instead show that education may induce financial institutions to renew obfuscation more frequently, undoing consumer learning and reducing welfare. Our mechanism differs from both: the policy changes which consumers remain susceptible to exploitation and may thereby induce the firm to exploit them more aggressively.} The following example demonstrates this possibility.\\


\begin{example} 
Suppose that $u(z)=z$, $v(z)=0$, and $c(z)=(1/2)z$, and that the policy induces the family of cognitive-cost distributions 
\[ F(\kappa;\tau) = \begin{cases} \displaystyle (1-\tau) \frac{\kappa^2} {\kappa^2+(1-\kappa)^2} +\tau\kappa, & 0\leq\kappa\leq\frac12,\\[1.2em] \displaystyle \frac{\kappa^2} {\kappa^2+(1-\kappa)^2}, & \frac12<\kappa\leq1, \end{cases} \qquad \tau\in[0,1]. \]
The policy is effective: $F(\kappa;\tau)$ is non-decreasing in $\tau$ for every $\kappa$, so stronger policy induces first-order stochastically lower cognitive costs. As Figure~\ref{fig:badcosts} shows, however, the policy is not progressive: it lowers cognitive costs among consumers in the lower half of the distribution while leaving the costs of those in the upper half unchanged. The remaining naive population consequently becomes less responsive to exploitation, inducing the monopolist to increase $x^*(\tau)$ which can also lead to decreases in consumer welfare. 
\end{example}

\begin{figure}[ht!]
    \centering

\begin{tikzpicture}

\pgfmathsetmacro{\xminstar}{0.396608}

\begin{groupplot}[
    group style={
        group size=3 by 1,
        horizontal sep=1.25cm
    },
    width=0.30\textwidth,
    height=0.40\textwidth,
    axis line style={black},
    tick style={black},
    title style={font=\small},
    label style={font=\small},
    tick label style={font=\scriptsize}
]


\nextgroupplot[
    xlabel={$\kappa$},
    title={$F(\kappa;\tau)$},
    xmin=0,
    xmax=1,
    ymin=0,
    ymax=1.02,
    xtick={0,0.5,1},
    ytick={0,0.2,0.4,0.6,0.8,1},
    legend style={
        font=\scriptsize,
        draw=none,
        fill=none,
        at={(0.03,0.97)},
        anchor=north west,
        row sep=-2pt
    }
]

\addplot[
    black,
    solid,
    thick,
    domain=0:1,
    samples=200
]
{x^2/(x^2+(1-x)^2)};
\addlegendentry{$\tau=0$}

\addplot[
    black,
    dashed,
    thick,
    domain=0:0.5,
    samples=100
]
{0.5*x^2/(x^2+(1-x)^2)+0.5*x};
\addlegendentry{$\tau=0.5$}

\addplot[
    black,
    dashed,
    thick,
    domain=0.5:1,
    samples=100,
    forget plot
]
{x^2/(x^2+(1-x)^2)};

\addplot[
    black,
    dotted,
    very thick,
    domain=0:0.5,
    samples=100
]
{x};
\addlegendentry{$\tau=1$}

\addplot[
    black,
    dotted,
    very thick,
    domain=0.5:1,
    samples=100,
    forget plot
]
{x^2/(x^2+(1-x)^2)};

\nextgroupplot[
    xlabel={$\tau$},
    title={$x^*(\tau)$},
    xmin=0,
    xmax=1,
    ymin=0.39,
    ymax=0.51,
    xtick={0,0.2,0.4,0.6,0.8,1},
    ytick={0.40,0.45,0.50}
]

\addplot[
    black,
    solid,
    thick,
    parametric,
    variable=\s,
    domain=\xminstar:0.5,
    samples=250
]
({
    (2*\s^4-4*\s^3+5*\s^2-4*\s+1)
    /(8*\s^5-18*\s^4+20*\s^3-11*\s^2+2*\s)
},
{\s});

\nextgroupplot[
    xlabel={$\tau$},
    title={$\mathrm{CS}(\tau)$},
    xmin=0,
    xmax=1,
    ymin=-0.377,
    ymax=-0.355,
    xtick={0,0.2,0.4,0.6,0.8,1},
    ytick={-0.375,-0.370,-0.365,-0.360}
]

\addplot[
    black,
    thick,
    parametric,
    variable=\s,
    domain=\xminstar:0.5,
    samples=250
]
({
    (2*\s^4-4*\s^3+5*\s^2-4*\s+1)
    /(8*\s^5-18*\s^4+20*\s^3-11*\s^2+2*\s)
},
{
    (
        (
            (2*\s^4-4*\s^3+5*\s^2-4*\s+1)
            /(8*\s^5-18*\s^4+20*\s^3-11*\s^2+2*\s)
        )*\s^2/2
    )
    -
    (
        (
            1+
            (2*\s^4-4*\s^3+5*\s^2-4*\s+1)
            /(8*\s^5-18*\s^4+20*\s^3-11*\s^2+2*\s)
        )*\s/2
    )
    -
    (
        (
            (2*\s^4-4*\s^3+5*\s^2-4*\s+1)
            /(8*\s^5-18*\s^4+20*\s^3-11*\s^2+2*\s)
            -1
        )/4
    )*ln(2*\s^2-2*\s+1)
});

\end{groupplot}
\end{tikzpicture}

    \caption{A non-progressive sophistication-promoting policy. The policy
    weakens the responsiveness of the remaining naive population, inducing
    greater exploitation and lower consumer welfare.}
    \label{fig:badcosts}
\end{figure}
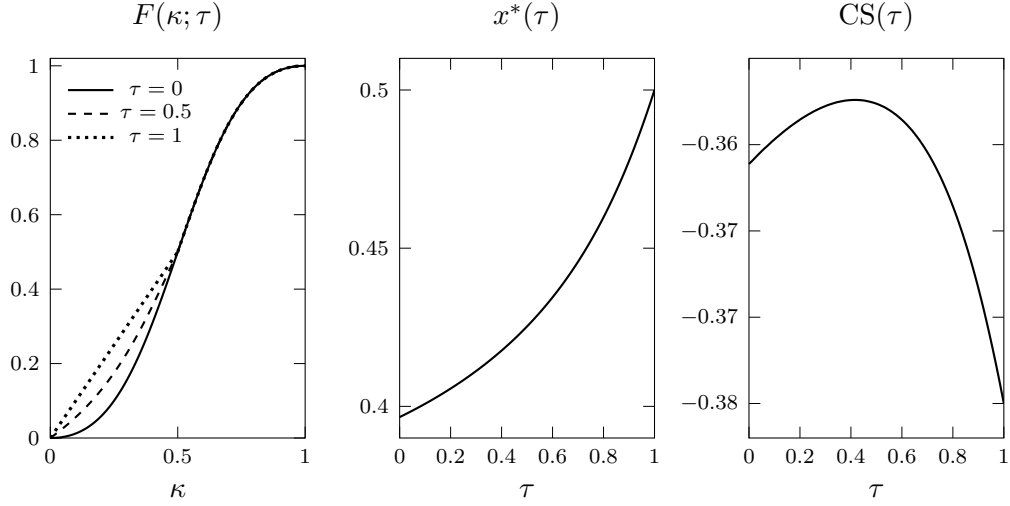

The same equilibrium responses that weaken the link between policy effectiveness and consumer welfare also complicate program evaluation; that is, the process of determining whether a given intervention is effective. In the exogenous benchmark, program evaluation is simple in principle: one need only determine whether the intervention increases the fraction of sophisticated consumers. With endogenous cognition, however, the equilibrium fraction of sophisticated consumers is $\eta^*(\tau)=F\bigl(x^*(\tau);\tau\bigr)$, which reflects both the policy's effect on cognitive costs and the equilibrium response of exploitation. Under an effective and progressive policy, these effects work in opposite directions: lower cognitive costs directly promote sophistication, whereas lower exploitation makes costly cognitive investment less necessary. Consequently, observed sophistication may fall in response to an effective and welfare-improving policy. If the endogeneity of cognitive states is ignored, such a beneficial intervention may therefore be misclassified as ineffective. The following example illustrates this possibility.

\begin{example}
Suppose that $u(z)=z$, $v(z)=0$, and $c(z)=1/2 z$, and that the policy induces a family of cognitive-cost distributions with support $[0,1]$ given by
\[ F(\kappa;\tau) = 1-(1-\kappa)e^{-\tau\kappa^2}, \qquad \kappa,\tau\in[0,1]. \]
The policy is effective and progressive. As shown in Figure~\ref{fig:eta}, stronger policy reduces the equilibrium depth of exploitation and increases consumer welfare. The reduction in exploitation is sufficiently large, however, so that the equilibrium fraction of sophisticated consumers decreases with policy intensity. 
\end{example}  

\begin{figure}[ht!]
    \centering
    \begin{tikzpicture}
\begin{groupplot}[
    group style={
        group size=4 by 1,
        horizontal sep=1.6cm
    },
    width=0.27\textwidth,
    height=5.8cm,
    tick label style={font=\footnotesize},
    label style={font=\small},
    title style={font=\small},
    xtick={0, 0.5, 1},
    xticklabels={$0$, $0.5$, $1$},
    every axis plot/.append style={thick}
]
 
\nextgroupplot[
    xmin=0, xmax=1,
    ymin=0, ymax=1,
    ytick={0, 0.5, 1},
    xlabel={$\kappa$},
    title={$F(\kappa;\tau)$}
]
\addplot[black, solid]  table {cdf_tau_c_halfz_1.dat};
\addplot[black, dashed] table {cdf_tau_c_halfz_0p5.dat};
\addplot[black, dotted] table {cdf_tau_c_halfz_0.dat};

\node[font=\scriptsize, anchor=west] at (axis cs:0.08, 0.90) {$\tau=1$};
\node[font=\scriptsize, anchor=west] at (axis cs:0.08, 0.80) {$\tau=0.5$};
\node[font=\scriptsize, anchor=west] at (axis cs:0.08, 0.68) {$\tau=0$};

\nextgroupplot[
    xmin=0, xmax=1,
    xlabel={$\tau$},
    title={$x^{*}(\tau)$}
]
\addplot[black] table {xstar_tau_c_halfz.dat};

\nextgroupplot[
    xmin=0, xmax=1,
    xlabel={$\tau$},
    title={$\mathrm{CS}(\tau)$}
]
\addplot[black] table {cs_tau_c_halfz.dat};
 
\nextgroupplot[
    xmin=0, xmax=1,
    xlabel={$\tau$},
    title={$F(x^{*}(\tau);\tau)$}
]
\addplot[black] table {Fxstar_tau_c_halfz.dat};

\end{groupplot}
\end{tikzpicture}

    \caption{A progressive sophistication-promoting policy. Exploitation falls
    and consumer welfare rises, while equilibrium sophistication decreases.}
    \label{fig:eta}
\end{figure}
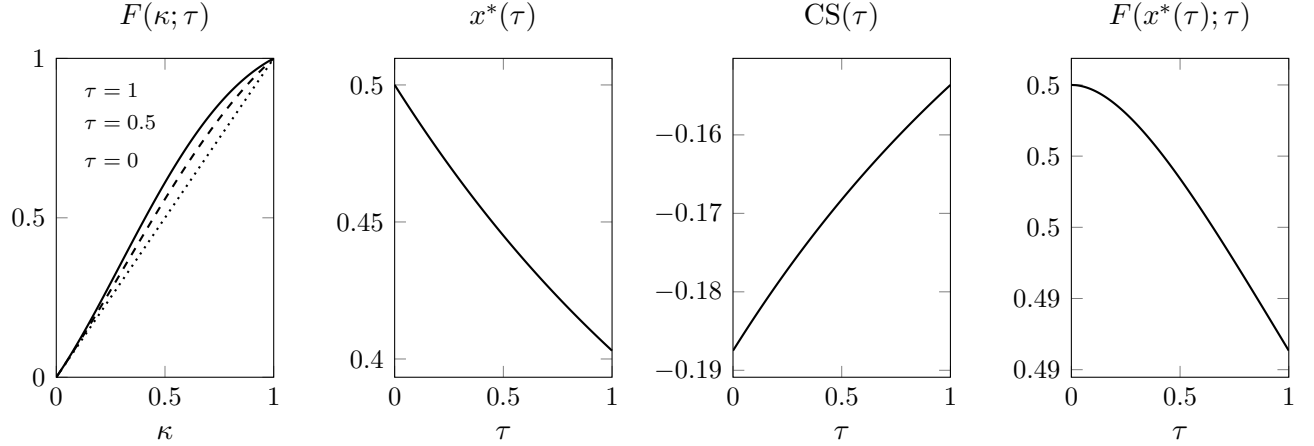

\subsubsection{Ex-Post Sophistication-Promoting Policies}\label{ssec:EP}

Ex-post policies improve the mapping from consumers' cognitive-investment decisions into their eventual cognitive states. Such interventions operate after consumers decide whether to invest in cognition but before their contract choices become final. Examples include salient disclosures, compliance checks, and cooling-off periods.
We model this improved mapping in a reduced form as follows: a consumer who does not invest in cognition nevertheless becomes sophisticated with probability $\tau\in[0,1]$. Thus, a higher $\tau$ represents a stronger and more effective ex-post sophistication-promoting policy.\\

For a given depth of exploitation $x$, a consumer who invests in cognition becomes sophisticated with certainty, whereas a consumer who does not invest remains naive with probability $1-\tau$. As such, the benefit of cognitive investment becomes $(1-\tau)x$. Thus the effective fraction of sophisticates is $\eta(x) \equiv F((1-\tau)x)+\tau (1-F((1-\tau)x)$, and the monopolist chooses the optimal depth of exploitation to solve \eqref{eqn:optimal_x} with $F(x)$ replaced with $\eta(x)$. A stronger policy reduces this benefit and crowds out private cognitive investment.\footnote{Related crowding out of private cognition appears in \citet{ArmstrongVickersZhou2009}: a price cap compresses price dispersion, reducing consumers' return to acquiring price information and thereby weakening price competition. In \citet{Grubb2015}, bill-shock alerts directly supply information about whether a penalty applies that consumers could otherwise obtain through costly attention to past usage. Our mechanism differs in that reduced cognitive investment reduces the naive-consumer base and may thereby induce deeper exploitation.} We call this response \emph{cognitive moral hazard}: consumers have less reason to invest in cognition since the policy may protect them anyway.\footnote{A real-world analogue of this behavioral response appears in \citet{patel2021}, who find that introducing free returns increased both product variety per order and subsequent returns, consistent with ex-post protection encouraging less selective ex-ante purchasing.} The following proposition identifies a condition under which cognitive moral hazard induces greater equilibrium exploitation.

\begin{prop}\label{prop:expost_result} 
Suppose that $F$ admits density $f$ and $\frac{\kappa f(\kappa)}{1-F(\kappa)}$ is increasing in $\kappa$.\footnote{The condition that \(\kappa f(\kappa)/[1-F(\kappa)]\) is increasing is the standard \emph{increasing generalized failure rate} (IGFR) condition. It is weaker than an increasing hazard rate: every distribution with an increasing hazard rate is IGFR, while some distributions with decreasing hazard rates are also IGFR. The condition is widely used in pricing and supply-chain applications; see \citet{Lariviere2006}.} Then, as the ex-post sophistication-promoting policy intensifies (i.e., as $\tau$ increases), the equilibrium depth of exploitation, $x^*$, increases. 
\end{prop}

The condition in Proposition~\ref{prop:expost_result} concerns the elasticity of the population remaining naive with respect to the return to cognitive investment. A stronger ex-post policy lowers this return. When the elasticity is increasing, this decrease moves the cognitive-investment threshold into a region where the naive population is less responsive to exploitation, weakening the discipline imposed on the monopolist. The monopolist consequently increases the equilibrium depth of exploitation.\footnote{This mechanism is related to the progressive ex-ante policies considered above. From the monopolist's perspective, an ex-post policy replaces \(1-F(x)\) with \(1-F((1-\tau)x)\); under the IGFR condition, this transformation thickens the relative upper tail, whereas a progressive ex-ante policy thins it. Accordingly, progressivity guarantees that exploitation decreases, while IGFR guarantees that exploitation increases.}\\

An increase in exploitation, however, does not imply that consumer welfare necessarily falls, because the ex-post policy also directly protects a larger fraction of consumers. Welfare falls only if the increase in \(x^*\) is sufficiently large to outweigh this direct protection, so that consumers who remain naive face greater effective exposure to exploitation. To see this, note that aggregate consumer welfare can be written as 
\[ 
CS(\tau) = -\int_0^{(1-\tau)x^*(\tau)} \bigl[1-F(\kappa)\bigr]\,d\kappa. 
\] 
Hence, $(1-\tau)x^*(\tau)$ is a sufficient statistic for consumer welfare: consumer welfare increases if $(1-\tau)x^*(\tau)$ decreases. If exploitation increases enough to outweigh the decline in $1-\tau$, then $(1-\tau)x^*(\tau)$ rises and consumer welfare can fall. The following example demonstrates this possibility.

\begin{example} Suppose $\kappa\sim U[0,1]$, $u(z)=z$, $v(z)=0$, and \[ c(z)= \begin{cases} \tfrac{z}{2}, & z<\tfrac{3}{4},\\[2pt] \tfrac{4}{5}z-\tfrac{9}{40}, & z\geq\tfrac{3}{4}. \end{cases} \] As Figure~\ref{fig:expost} shows, for low values of $\tau$ the policy backfires: $(1-\tau)x^*(\tau)$ increases, so consumer welfare falls despite the direct protective effect of the intervention. Once $\tau$ exceeds $\tau^*=\tfrac{23}{43}$, exploitation reaches the corner $\bar{x}=1$ and consumer welfare recovers. \end{example}
\begin{figure}[ht!]
\centering

\pgfplotsset{
  expost/.style={
    width        = 0.47\textwidth,
    height       = 0.36\textwidth,
    xmin         = 0,
    xmax         = 1,
    xtick        = {0, 0.2, 23/43, 0.8, 1},
    xticklabels  = {$0$, $0.2$, $\tau^*$, $0.8$, $1$},
    xlabel       = {$\tau$},
    axis lines   = left,
    axis line style   = {thin, black},
    tick style        = {thin, black},
    tick align        = outside,
    ticklabel style   = {font=\small},
    label style       = {font=\small},
    legend cell align = left,
    legend style      = {
      font      = \small,
      draw      = none,
      fill      = none,
      row sep   = 2pt,
    },
  }
}
 
\begin{tikzpicture}
\begin{axis}[
  expost,
  ymin  = 0,
  ymax  = 1.05,
  ytick = {0, 0.25, 0.5, 0.75, 1},
  yticklabels = {$0$, $0.25$, $0.5$, $0.75$, $1$},
  legend style = {
    font         = \small,
    draw         = none,
    fill         = none,
    fill opacity = 0.85,
    text opacity = 1,
    at           = {(0.05,1)},
    anchor       = north west,
    row sep      = 1pt,
  },
]

  \addplot[
    red!75!black,
    densely dashed,
    domain=0:1,
    samples=500,
  ]
  {
    ifthenelse(
      x < 23/43,
      17/40 + 3*x/40,
      1-x
    )
  };
  \addlegendentry{$(1{-}\tau)x^*(\tau)$}

  \addplot[
    blue!80!black,
    domain=0:1,
    samples=500,
  ]
  {
    ifthenelse(
      x < 23/43,
      1/(2*(1-x)) - 3/40,
      1
    )
  };
  \addlegendentry{$x^*(\tau)$}

  \addplot[
    gray!60,
    dotted,
    line width=0.8pt,
  ]
  coordinates {
    (23/43,0)
    (23/43,1.05)
  };

\end{axis}
\end{tikzpicture}%
\hfill%
\begin{tikzpicture}
\begin{axis}[
  expost,
  ymin  = -0.38,
  ymax  = 0.01,
  ytick = {0, -0.1, -0.2, -0.3},
  yticklabels = {$0$, $-0.1$, $-0.2$, $-0.3$},
  legend style = {
    font         = \small,
    draw         = none,
    fill         = white,
    fill opacity = 0.85,
    text opacity = 1,
    at           = {(0.05,0.98)},
    anchor       = north west,
    row sep      = 1pt,
  },
]

  \addplot[
    blue!80!black,
    domain=0:1,
    samples=500,
  ]
  {
    ifthenelse(
      x < 23/43,
      -(17/40 + 3*x/40)
        + 0.5*(17/40 + 3*x/40)^2,
      -(1-x) + 0.5*(1-x)^2
    )
  };
  \addlegendentry{$CS(\tau)$}

  \addplot[
    gray!60,
    dotted,
    line width=0.8pt,
  ]
  coordinates {
    (23/43,-0.38)
    (23/43,0.01)
  };

\end{axis}
\end{tikzpicture}

\caption{%
  Ex-post sophistication policy. The left panel depicts equilibrium
  exploitation $x^*(\tau)$ (solid) and the sufficient statistic
  $(1-\tau)x^*(\tau)$ (dashed). The right panel depicts consumer
  surplus $CS(\tau)$. The vertical lines mark the threshold
  $\tau^*=\tfrac{23}{43}$. For $\tau<\tau^*$, the policy backfires:
  exploitation rises and consumer surplus falls. Once exploitation
  reaches the corner $\bar{x}=1$, consumer surplus recovers.
}
\label{fig:expost}
\end{figure}
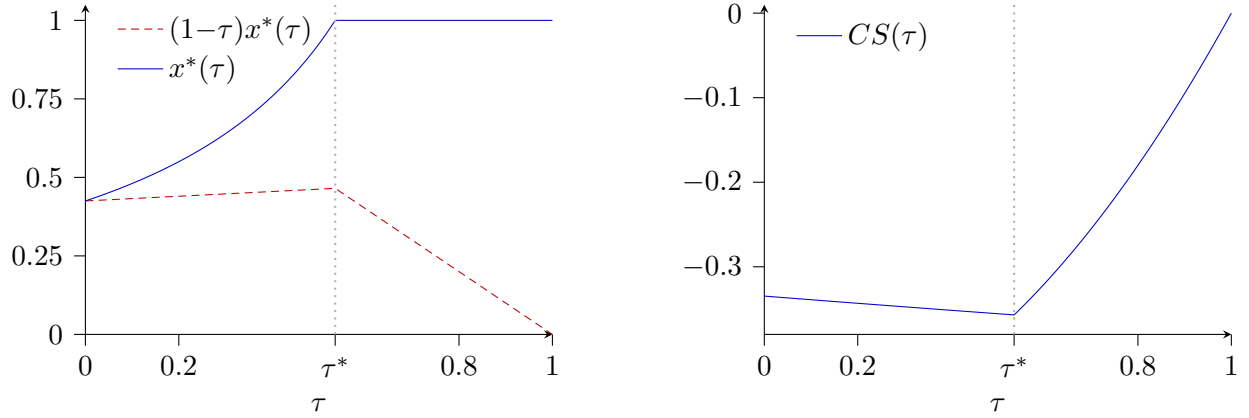

A natural question is whether an effective ex-post sophistication-promoting policy can be guaranteed to increase consumer welfare. The answer depends on how strongly excess profits respond to deeper exploitation. Recall that $x-L(x)$ is the monopolist's excess profit from a naive consumer. Since $L$ is convex, $L'(x)$ exists almost everywhere. At every point of differentiability for which $x-L(x)>0$, define 
\[ 
\mathcal{E}(x) \equiv \frac{x\bigl[1-L'(x)\bigr]}{x-L(x)} 
\] 
as the elasticity of excess profits with respect to exploitation. We say that excess profits have \emph{decreasing elasticity} if $\mathcal{E}(x)\geq\mathcal{E}(x')$ for any $x<x'$ at which $\mathcal{E}$ is well-defined.

\begin{prop}\label{prop:elasticity} 
Consumer welfare increases for every cognitive-cost distribution $F \in \mathcal{F}$ as an ex-post sophistication-promoting policy becomes more effective (i.e., as $\tau$ increases) if and only if excess profits have decreasing elasticity. 
\end{prop}

When $\mathcal{E}(x)$ is decreasing, the proportional return to deeper exploitation declines with $x$. The monopolist therefore does not increase exploitation sufficiently rapidly to offset the policy-induced reduction in the private return to cognitive investment. Consequently, $(1-\tau)x^*(\tau)$ decreases and the policy's direct protective effect dominates.

\subsection{Preference Nudges}\label{subsec:PN}

The final friction underlying exploitation in our benchmark is consumers' dynamic inconsistency itself. Preference nudges target this friction by attenuating the conflict between consumers' long-run objectives and short-run temptation. Examples include defaults, reminders, and commitment prompts that help consumers act according to their long-run preferences.\footnote{Defaults and reminders are discussed in \cite{thaler2008nudge}, while \cite{AllcottSunstein2015} study policies designed to correct `internalities'.} We model such policies as replacing temptation utility $v$ with long-run utility $u$ at the outcome-selection stage for a fraction $\tau\in[0,1]$ of consumers. The nudge therefore removes these consumers' preference conflict and makes their choices dynamically consistent. Because their short-run behavior is now aligned with their long-run preferences, their cognitive states are irrelevant. The remaining fraction $1-\tau$ continues to choose according to $v$ (i.e., remain dynamically inconsistent), so their cognitive investment determines whether these consumers correctly anticipate their subsequent behavior.\footnote{Allowing $\tau<1$ accommodates heterogeneity in responsiveness to preference nudges. It also allows dynamically consistent and dynamically inconsistent consumers to coexist in the market, showing that the model's central mechanism does not depend on all consumers being dynamically inconsistent.} Thus, a higher $\tau$ represents a stronger and more effective preference nudge.\\

Preference nudges alter the monopolist's contracting problem in two related ways. First, dynamically consistent and naive consumers are behaviorally indistinguishable at the contract-selection stage because both evaluate contracts according to $u$. The monopolist therefore cannot screen these consumers through their initial contract choices. Second, the $u$-optimal outcome offered by the contract becomes relevant for realized profits: dynamically consistent consumers subsequently select this outcome, whereas naive consumers instead choose according to $v$. These features apply under both exogenous and endogenous cognition, although with different implications for exploitation.\\

We first consider the benchmark in which cognitive states are exogenous among
dynamically inconsistent consumers. In this case, the monopolist's problem is 
\begin{equation}\label{eqn:nudge}
    \max \limits_{\{\sigma^S, \sigma^N\} \in \mathcal{M}} (1 - \tau)\bar{\eta} \pi(y_v^S) + \tau \pi(y_u^N) + (1 - \tau)(1 - \bar{\eta}) \pi(y_v^N).
\end{equation}

The following proposition establishes that the presence of dynamically consistent consumers only serve to discipline the monopolist in this case.

\begin{prop}\label{prop:exoN} 
Suppose that a fraction $\tau$ of consumers is dynamically consistent and that cognitive states are exogenous among the remaining $1-\tau$ dynamically inconsistent consumers, of whom a fraction $\bar{\eta}\in[0,1]$ is sophisticated. Then, consumer welfare increases as preference nudges become more effective (i.e., as $\tau$ increases). 
\end{prop}

Recall that, when \(\tau=0\), the contract's \(u\)-optimal outcome serves only as bait because naive consumers ultimately choose according to \(v\). As the fraction of dynamically consistent consumers increases, however, more consumers choose this outcome, giving the bait allocation direct welfare and profit consequences. The monopolist therefore places greater weight on its efficiency, so maintaining any given depth of exploitation requires greater allocative distortions elsewhere in the contract. Because distortions become increasingly costly at the margin, further exploitation becomes less profitable, inducing the monopolist to reduce the depth of exploitation. When cognitive states are exogenous, this within-contract distortion is the only equilibrium force through which preference nudges affect the monopolist's choice of exploitation.\\

With endogenous cognition, the within-contract distortion is still present and its effect on the depth--reach tradeoff must be accounted for.\footnote{In this case, the problem is the same as in \eqref{eqn:nudge}, but with $\bar{\eta}$ replaced with $F(x)$.} Because interacting with a naive consumer is more profitable than with a sophisticated one, the monopolist has an incentive to preserve naivet\'e and this is precisely what limits its desire to exploit. By increasing within-contract distortions, a stronger preference nudge erodes this excess profit and thereby weakens that limiting force.\footnote{\citet{MurookaSchwarz2019} identify a related force in an automatic-renewal market. A policy that facilitates switching disproportionately drives rational consumers away, reducing the firm's return from serving them and making the exploitation of naive consumers relatively more attractive. In our model, the policy instead reduces the excess profit from preserving consumers' naivet\'e, which may induce the monopolist to exploit the remaining naive consumers more aggressively.} Thus, more effective preference nudges may actually increase exploitation.\footnote{\citet{Spiegler2015} likewise emphasizes that firms' equilibrium responses can reverse the intended effects of behavioral nudges. His reversals arise from separate procedural models of default bias, mistaken predictions of future usage, and limited attention, rather than from endogenous sophistication or the depth--reach feedback studied here.}\\

To formalize this trade-off, let $\Delta(\tau)$ denote the additional allocative distortion required to implement a given depth of exploitation when a fraction $\tau$ of consumers is dynamically consistent, where $\Delta(\tau)$ is increasing in $\tau$. For a given depth of exploitation $x$, the monopolist's excess profit from a naive dynamically-inconsistent consumer is $x-L\bigl(x+\Delta(\tau)\bigr)$. A stronger preference nudge has two opposing effects on the monopolist's incentive to exploit consumers. First, because $L$ is convex, the marginal return to deeper exploitation, $1-L'\bigl(x+\Delta(\tau)\bigr)$, falls as $\Delta(\tau)$ increases. Second, the level of excess profit, $x-L\bigl(x+\Delta(\tau)\bigr)$, also falls. The condition provided in the following proposition ensures that within-contract distortions dominate the erosion of the naive-consumer profit base, inducing the monopolist to reduce exploitation as policy strengthens.

\begin{prop}\label{prop:endoN}
Suppose that cognitive states are endogenous among the \(1-\tau\) fraction of
dynamically inconsistent consumers. Suppose further that, for every \(x\),
\begin{equation}\label{eqn:nudge_good}
    \frac{1-L'(x+\Delta)}
         {x-L(x+\Delta)}
\end{equation}
is decreasing in \(\Delta\) for all \(\Delta\geq 0\) for which the expression
is well defined. Then, as the fraction of dynamically consistent consumers, $\tau$, increases, the equilibrium depth of exploitation,
\(x^*(\tau)\), decreases, and consumer welfare increases for every cognitive-cost distribution \(F \in \mathcal{F}\).
\end{prop}

Thus, under the condition \eqref{eqn:nudge_good}, a more effective preference nudge increases consumer welfare through three channels: more consumers act according to their long-run preferences, dynamically inconsistent consumers who remain naive are exploited less, and this lower depth of exploitation reduces costly cognitive investment. When the condition fails, however, the monopolist exploits those who remain naive more aggressively, which also increases defensive and wasteful cognitive investment. This response can be strong enough to overturn the policy's direct benefits of promoting dynamic consistency, leading to a reduction in aggregate consumer welfare. The following example illustrates this possibility.

\begin{example} Suppose that \(\kappa\sim U[0,1]\), \(u(z)=z\), \(v(z)=0\), and \[ c(z)= \begin{cases} \frac{1}{2}z & z<\frac{1}{2},\\[3pt] \frac{3}{2}z-\frac{1}{2} & z\geq\frac{1}{2}. \end{cases} \] Maximal long-run allocative welfare is \(W^*=\frac14\), attained at \(z=\frac12\), and the associated allocative-loss function is \(L(x)=0\) for \(x\leq\frac12\) and \(L(x)=\frac12x-\frac14\) for \(x>\frac12\). Writing \(\Delta=1-z_u\), the monopolist chooses \(\Delta\) to maximize \(\tau A(\Delta)+(1-\tau)V(\Delta)\), where \(A(\Delta)=\frac12\min\{\Delta,1-\Delta\}\) and \(V(\Delta)=\max_x(1-x)[x-L(x+\Delta)]\). Before $\tau^* \equiv (4-\sqrt{2})/7$, a stronger nudge gradually changes the bait outcome and reduces exploitation. 
As \(\tau\) crosses \(\tau^*\), the monopolist switches to \(\Delta^*(\tau)=\frac12\), causing exploitation to jump from \(\sqrt{2}/4\) to \(1/2\). After the threshold, exploitation remains at \(x^*(\tau)=1/2\), while consumer welfare increases as the fraction of dynamically inconsistent consumers shrinks. \end{example}
  
     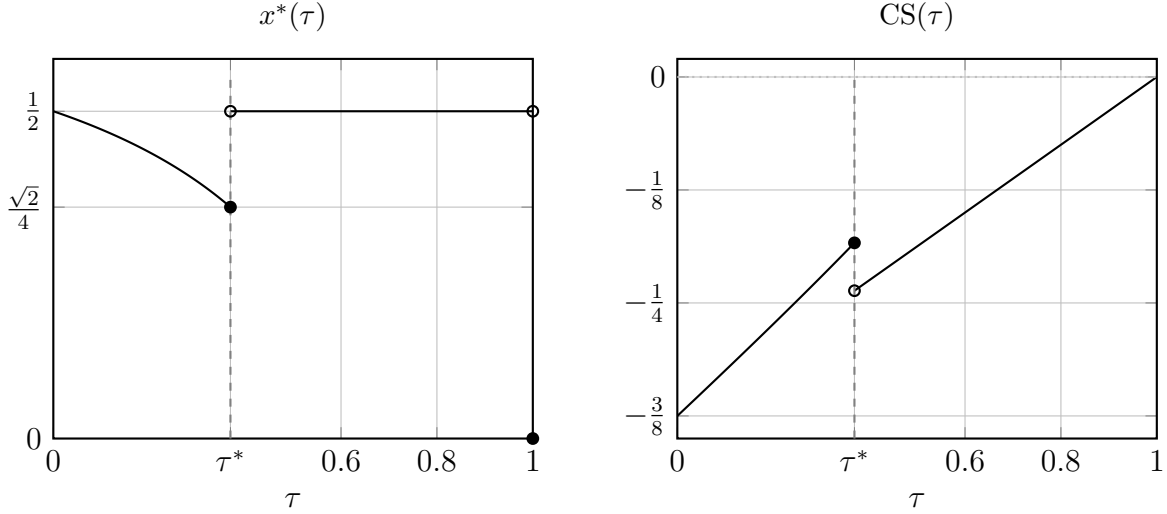
\begin{figure}[ht!]
\centering

\begin{tikzpicture}

\pgfmathsetmacro{\taustar}{(4-sqrt(2))/7}
\pgfmathsetmacro{\xminus}{sqrt(2)/4}
\pgfmathsetmacro{\csminus}{-(5+11*sqrt(2))/112}
\pgfmathsetmacro{\csplus}{-3*(3+sqrt(2))/56}

\begin{axis}[
    width=0.48\textwidth,
    height=0.40\textwidth,
    at={(0,0)},
    anchor=south west,
    xlabel={$\tau$},
    ylabel={},
    title={$x^*(\tau)$},
    title style={font=\small},
    xmin=0,
    xmax=1,
    ymin=0,
    ymax=0.58,
    ytick={0,{sqrt(2)/4},0.5},
    yticklabels={$0$,$\frac{\sqrt{2}}{4}$,$\frac{1}{2}$},
    xtick={0,{(4-sqrt(2))/7},0.6,0.8,1},
    xticklabels={$0$,$\tau^*$,$0.6$,$0.8$,$1$},
    grid=both,
    thick
]

\addplot[
    black,
    thick,
    domain=0:\taustar,
    samples=150
]
{0.5-x/(4*(1-x))};

\addplot[
    black,
    thick,
    domain=\taustar:1,
    samples=2
]
{0.5};

\draw[dashed, gray]
    (axis cs:\taustar,0)
    --
    (axis cs:\taustar,0.55);

\addplot[
    black,
    only marks,
    mark=*,
    mark size=2pt
]
coordinates {(\taustar,\xminus)};

\addplot[
    black,
    only marks,
    mark=o,
    mark size=2pt,
    fill=white
]
coordinates {(\taustar,0.5)};

\addplot[
    black,
    only marks,
    mark=o,
    mark size=2pt,
    fill=white
]
coordinates {(1,0.5)};

\addplot[
    black,
    only marks,
    mark=*,
    mark size=2pt
]
coordinates {(1,0)};

\end{axis}

\begin{axis}[
    width=0.48\textwidth,
    height=0.40\textwidth,
    at={(0.50\textwidth,0)},
    anchor=south west,
    xlabel={$\tau$},
    ylabel={},
    title={$\mathrm{CS}(\tau)$},
    title style={font=\small},
    xmin=0,
    xmax=1,
    ymin=-0.40,
    ymax=0.02,
    ytick={-0.375,-0.25,-0.125,0},
    yticklabels={$-\frac{3}{8}$,$-\frac{1}{4}$,$-\frac{1}{8}$,$0$},
    xtick={0,{(4-sqrt(2))/7},0.6,0.8,1},
    xticklabels={$0$,$\tau^*$,$0.6$,$0.8$,$1$},
    grid=both,
    thick
]

\addplot[
    black,
    thick,
    domain=0:\taustar,
    samples=150
]
{
-(1-x)*
(
    (0.5-x/(4*(1-x)))
    -0.5*(0.5-x/(4*(1-x)))^2
)
};

\addplot[
    black,
    thick,
    domain=\taustar:1,
    samples=150
]
{-3/8*(1-x)};

\draw[dashed, gray]
    (axis cs:\taustar,-0.40)
    --
    (axis cs:\taustar,0);

\draw[dotted, lightgray]
    (axis cs:0,0)
    --
    (axis cs:1,0);

\addplot[
    black,
    only marks,
    mark=*,
    mark size=2pt
]
coordinates {(\taustar,\csminus)};

\addplot[
    black,
    only marks,
    mark=o,
    mark size=2pt,
    fill=white
]
coordinates {(\taustar,\csplus)};

\end{axis}

\end{tikzpicture}

\caption{Preference nudges, equilibrium exploitation, and consumer welfare. The left panel plots the equilibrium depth of exploitation $x^*(\tau)$. The right panel depicts consumer welfare $\mathrm{CS}(\tau)$ as the nudge intensity $\tau$ increases.}  
\label{fig:nudge}
\end{figure}

\section{Conclusion}

We develop a general contracting framework in which consumers can invest costly resources to recognize and avoid their behavioral mistakes. Because the return to sophistication is determined by the consequences of remaining naive, firms' contracts affect not only how deeply naive consumers are exploited but also how many consumers remain susceptible to exploitation. Firms therefore face a trade-off between the depth and reach of exploitation: deeper exploitation raises profit from each consumer who remains naive, but induces greater cognitive investment and reduces the exploitable population. We show that this trade-off disciplines exploitation in the unregulated market and creates a common channel through which competition, sophistication-promoting interventions, and preference nudges can produce unintended consequences.\\

The central lesson for policy evaluation is that neither the depth nor the reach of exploitation can be interpreted in isolation. A reduction in reach need not indicate successful de-biasing. It may instead arise because deeper exploitation induces more consumers to undertake costly defensive investment in sophistication. Conversely, an increase in reach need not indicate policy failure. A policy may reduce exploitation so substantially that cognitive investment becomes less valuable, causing more consumers to remain naive even as their welfare improves. In other words, observed sophistication reflects both consumers' cognitive investment and the exploitation against which cognition protects them. Changes in the cognitive composition of a market are informative about consumer outcomes only when interpreted jointly with changes in exploitation depth.\\

The same logic changes welfare accounting. Cognitive investment consumes real resources and, in our setting, is undertaken defensively to avoid exploitation. A market with many sophisticated consumers may therefore perform poorly precisely because consumers must expend substantial resources protecting themselves. Policy evaluation must account jointly for the losses imposed on consumers who remain naive, the reach of those losses, and the cognitive costs incurred by consumers who avoid them. More generally, an intervention that reduces one margin of exploitation may weaken the discipline operating through the other. Thus, the relevant policy question is not simply whether exploitation reaches more or fewer consumers, but how the intervention changes the relationship between depth and reach and, through that relationship, firms' incentives to exploit.\\

Although we develop the framework for dynamically inconsistent consumers, its central logic extends to other behavioral markets. Existing work on endogenous cognition largely concerns information about external objects, such as prices, product characteristics, valuations, or economic states. An important direction for future research is to study endogenous awareness of internal determinants of behavior. Consumers may invest in understanding whether present bias, projection bias, inertia, salience, loss aversion, or other behavioral forces will govern their choices when a market interaction unfolds. Because firms design the contracts, products, and choice environments in which these forces become consequential, market conduct may determine both the value of awareness and the population that remains unaware. The depth--reach trade-off may therefore provide a common framework for studying markets in which firms influence not only the consequences of consumer mistakes, but also consumers' incentives to recognize them.

\bibliographystyle{jpe}
\bibliography{cleaned_references}

\appendix 
\section{Cognitive Equilibrium and the Revelation Principle} \label{app:cognitive_equilibrium}

In this appendix we define the solution concept \emph{cognitive equilibrium}
and establish that it is without loss of generality to restrict attention to
direct contracts and directed menus. Throughout, ties are resolved according
to the equilibrium selection under consideration. Equivalently, all statements
may be read as applying to a fixed measurable selection from the relevant
argmax correspondences.\\

Consider a two-period game played by a consumer. In period $t=1$ the
consumer chooses an action $a^{1}\in\mathcal A_{1}$. In period $t=2$ the
consumer chooses $a^{2}\in\mathcal A_{2}(a^{1})$, where the feasible set may
depend on the period-1 action. The profile $(a^{1},a^{2})$ determines a
distribution over outcomes $\rho(a^{1},a^{2})\in\Delta(Z\times\mathbb R)$.
For $f\in\{u,v\}$, profile $(a^1,a^2)$ is evaluated as
\[
J_f(a^1,a^2)
\equiv
\int_{Z\times\mathbb R}[f(z)-p] \,d\rho(a^1,a^2).
\]

\begin{definition}[Cognitive Equilibrium]
Fix a two-period game $(U,V,F,\rho,\mathcal A_1,\mathcal A_2)$. A
\emph{cognitive equilibrium} consists of period-2 choice rules
$a_f^2(a^1)$ for $f\in\{u,v\}$, period-1 choices $a_S^1,a_N^1$, and a
fraction of sophisticated consumers $\eta$ such that:
\begin{enumerate}
\item for every $f\in\{u,v\}$ and $a^1\in\mathcal A_1$,
\[
a_f^2(a^1)\in \arg\max_{a^2\in\mathcal A_2(a^1)}J_f(a^1,a^2);
\]
\item
\[
a_S^1\in\arg\max_{a^1\in\mathcal A_1}J_u(a^1,a_v^2(a^1)),
\qquad
a_N^1\in\arg\max_{a^1\in\mathcal A_1}J_u(a^1,a_u^2(a^1));
\]
\item the equilibrium fraction of sophisticated consumers is $\eta=F(x)$, where
\[
x
=
J_u(a_S^1,a_v^2(a_S^1))-J_u(a_N^1,a_v^2(a_N^1)),
\]
defines the depth of exploitation.
\end{enumerate}
Sophisticated consumers thus anticipate their actual $V$-maximizing
continuation choice, whereas naive consumers forecast the continuation choice
that would maximize $U$.
\end{definition}

\begin{lemma}[Revelation Principle]\label{lem:revelation}
For every cognitive equilibrium of any two-period game, there exists a directed
menu of direct contracts that generates the same type-contingent distributions
over realized outcomes and the same value of sophistication $x$.
\end{lemma}

\begin{proof}
Fix a cognitive equilibrium. Let
\[
y_v^S\equiv\rho(a_S^1,a_v^2(a_S^1)),
\qquad
y_u^N\equiv\rho(a_N^1,a_u^2(a_N^1)),
\qquad
y_v^N\equiv\rho(a_N^1,a_v^2(a_N^1)).
\]
The notation treats a distribution over primitive outcomes as a single
lottery-valued outcome, exactly as in the main text.\\

Construct the sophisticated contract as the singleton
$\sigma^S=\{y_v^S\}$. It is direct because the same outcome is selected under
$U$ and $V$. Construct the naive contract as
$\sigma^N=\{y_u^N,y_v^N\}$. By optimality of $a_u^2(a_N^1)$ under $U$ and
$a_v^2(a_N^1)$ under $V$,
\[
U(y_u^N)\ge U(y_v^N),
\qquad
V(y_v^N)\ge V(y_u^N),
\]
so $\sigma^N$ is direct.\\

It remains to verify incentive compatibility at the contract selection stage:  Because $a_N^1$ maximizes the naive
consumer's perceived long-run payoff, $U(y_u^N)
=J_u(a_N^1,a_u^2(a_N^1))
\ge J_u(a_S^1,a_u^2(a_S^1)) $. Since $a_u^2(a_S^1)$ maximizes $U$ following $a_S^1$,
$J_u(a_S^1,a_u^2(a_S^1))\ge U(y_v^S)$. Hence
$U(y_u^N)\ge U(y_v^S)$. Likewise, because $a_S^1$ maximizes the
sophisticated consumer's correctly anticipated long-run payoff, $U(y_v^S)
=J_u(a_S^1,a_v^2(a_S^1))
\ge J_u(a_N^1,a_v^2(a_N^1))
=U(y_v^N) $.\\

Participation follows from the availability of the outside option in the
original game. Therefore $\{\sigma^S,\sigma^N\}$ is a directed menu.
Sophisticated consumers select $\sigma^S$ and realize $y_v^S$; naive
consumers select $\sigma^N$, expect $y_u^N$, and realize $y_v^N$. The
replicated depth of exploitation is
$U(y_v^S)-U(y_v^N)$, which equals the original $x$ by construction.
\end{proof}

\section{Proofs of Section~\ref{sec:results}}
\begin{proof}[Proof of Lemma~\ref{lemma:contracts_exogenous}] We first characterize optimal contracts when the monopolist observes the consumer's cognitive state, and then show that the resulting contracts form a directed menu.
\begin{lemma}\label{lemma:hypo} Suppose the monopolist observes $\omega\in\{S,N\}$. The profit-maximizing realized outcome is $(z_u^*,u(z_u^*))$ for a sophisticated consumer. For a naive consumer it is \[ y_v^N = \left( z_v^*, v(z_v^*)+\max_{z\in Z}\{u(z)-v(z)\} \right), \] supported by a bait outcome $y_u^N=(z',u(z'))$, where $z'\in\arg\max_{z\in Z}\{u(z)-v(z)\}$. \end{lemma}

\begin{proof} For a sophisticated consumer, the monopolist may offer a single outcome. Individual rationality requires $p\leq u(z)$, so profit is at most $u(z)-c(z)$. It is maximized at $z_u^*$, with price $u(z_u^*)$. \\

For a naive consumer, write the two outcomes as $y_f^N=(\alpha_f^N,p_f^N)$, $f\in\{u,v\}$. The monopolist solves \[ \max_{\alpha_u^N,\alpha_v^N,p_u^N,p_v^N} \quad p_v^N-\int_0^1 c(z)\,d\alpha_v^N(z) \] subject to \[ U(y_u^N)\geq 0, \qquad U(y_u^N)\geq U(y_v^N), \qquad V(y_v^N)\geq V(y_u^N). \] At an optimum, perceived participation binds:  $ p_u^N = \int_0^1 u(z)\,d\alpha_u^N(z)$. The $V$-incentive constraint then implies $p_v^N \leq \int_0^1 v(z)\,d\alpha_v^N(z) + \int_0^1 \bigl[u(z)-v(z)\bigr]\,d\alpha_u^N(z)  $.  It binds because only $p_v^N$ enters realized profit. Consequently, profit is \[ \int_0^1 \bigl[v(z)-c(z)\bigr]\,d\alpha_v^N(z) + \int_0^1 \bigl[u(z)-v(z)\bigr]\,d\alpha_u^N(z), \] which is maximized by concentrating $\alpha_v^N$ on $z_v^*$ and $\alpha_u^N$ on a maximizer $z'$ of $u-v$. Substitution yields the stated prices. Finally, \[ U(y_u^N)-U(y_v^N) = \max_{z\in Z}\{u(z)-v(z)\} - \bigl[u(z_v^*)-v(z_v^*)\bigr] \geq 0, \] so the remaining incentive constraint is satisfied. \end{proof}

The two full-information contracts form a directed menu. Both intended contract-selection outcomes give perceived utility zero. Sophisticated consumers weakly prefer their contract because it gives realized long-run utility zero, while the naive contract gives realized long-run utility $-\bar{x}\leq 0$. Naive consumers perceive utility zero from their intended contract and no more than zero from the sophisticated contract. Thus, the menu attains the full-information upper bound and is optimal without observability. 
\end{proof}

\begin{proof}[Proof of Proposition~\ref{prop:exoiff}] Lemma~\ref{lemma:contracts_exogenous} gives $U(y_v^S)=0$ and \(   U(y_v^N) = u(z_v^*)-v(z_v^*) - \max_{z\in Z}\{u(z)-v(z)\} = -\bar{x}. \) Hence, the market is exploitative if and only if $\bar{\eta}<1$ and $\bar{x}>0$, equivalently, \(   z_v^* \notin \arg\max_{z\in Z}\{u(z)-v(z)\}. \) Indeed, if $z_v^*$ maximizes $u-v$, then it maximizes both $v-c$ and $u-v$, and therefore their sum $u-c$. Uniqueness of the maximizer of $u-c$ implies $z_v^*=z_u^*$. Conversely, if $z_v^*$ does not maximize $u-v$, the defining difference for $\bar{x}$ is strictly positive. \end{proof}

\begin{proof}[Proof of Lemma~\ref{lemma:contracts_endogenous} and Proposition~\ref{prop:sanity}] Fix an implementable $x$. The sophisticated contract is unchanged from the exogenous benchmark and generates profit \(   W^* \equiv \max_{z\in Z}\{u(z)-c(z)\}. \) Since $U(y_v^S)=0$, the  sophistication constraint requires $U(y_v^N)=-x$, or \(   p_v^N = \int_0^1 u(z)\,d\alpha_v^N(z)+x. \) \\

As before, perceived participation of the naive consumer binds, and the bait lottery can be concentrated on a maximizer of $u-v$. The $V$-choice constraint is then equivalent to \(   \max_{z\in Z}\{u(z)-v(z)\} - \int_0^1 \bigl[u(z)-v(z)\bigr]\,d\alpha_v^N(z) \geq x. \) The realized profit from the naive contract is \( x + \int_0^1 \bigl[u(z)-c(z)\bigr]\,d\alpha_v^N(z) = W^*+x-L(x). \) This proves Lemma~\ref{lemma:contracts_endogenous}. 
\end{proof}

\begin{proof}[Proof of Proposition~\ref{prop:sanity}]
Let \(R(x)\equiv x-L(x)\). Its equivalent maximization representation is
\[
W^*+R(x)
=
\max_{\alpha\in\Delta(Z)}
\left\{
x+\int_0^1 [u(z)-c(z)]\,d\alpha(z)
:
\max_{z\in Z}\{u(z)-v(z)\}
-\int_0^1 [u(z)-v(z)]\,d\alpha(z)
\geq x
\right\}.
\]
The objective and constraint are affine in \((x,\alpha)\), and lotteries
can be mixed. Mixing optimal lotteries therefore shows that \(R\) is
concave. Moreover, \(L(0)=0\), so \(R(0)=0\), and, as established
above, \(\bar{x}\) is the unique maximizer of \(R\).\\

Because \(\bar{x}>0\) uniquely maximizes \(R\), \( R(\bar{x})>R(0)=0 \). Concavity therefore implies that, for every \(x\in(0,\bar{x})\), \(R(x) > 0\). Hence, \( [1-F(x)]R(x)>0 \) for every \(x\in(0,\bar{x})\), whereas the monopolist's objective equals zero at \(x=0\). Thus, no equilibrium depth can equal zero.\\

Finally, \(R\) is decreasing to the right of its maximizer \(\bar{x}\),
while \(1-F\) is decreasing and nonnegative. Therefore, no
\(x>\bar{x}\) can outperform \(\bar{x}\).

\end{proof}

\section{Proofs for Section~\ref{sec:PA}}

\subsubsection*{C.1\quad Competition}

\begin{proof}[Proof of Lemma~\ref{lemma:contracts_comp}] We first show that, for each cognitive type present in the market, every contract selected with positive probability maximizes that type's perceived utility subject to nonnegative profit. Suppose instead that a type selects a contract $\sigma^\omega$ that does not solve this problem. Consider a stochastic contract that assigns probability $\alpha$ to $\sigma^\omega$ and probability $1-\alpha$ to a zero-profit contract that maximizes the type's perceived utility. For every $\alpha<1$, the stochastic contract gives the type strictly greater perceived utility than $\sigma^\omega$. After increasing its price by a sufficiently small amount, a rival can therefore make the contract strictly profitable while preserving this strict preference. \\

Under exogenous cognition, the deviating contract attracts all consumers of the relevant type. Under endogenous cognition, as $\alpha\to1$, the stochastic contract and its implied benefit of cognition converge to those in the candidate equilibrium. Continuity of $F$ therefore ensures that, for $\alpha$ sufficiently close to one, a positive mass of the relevant type remains and selects the deviating contract. The deviation consequently earns strictly positive profit, contradicting equilibrium. Hence, every contract selected with positive probability solves the relevant perceived-utility maximization problem subject to nonnegative profit. \\

For sophisticated consumers, perceived and realized contract values coincide. Their problem is to maximize $\int u(z)\,d\alpha(z)-p$ subject to $p\geq\int c(z)\,d\alpha(z)$. The profit constraint binds, leaving the objective $\int[u(z)-c(z)]\,d\alpha(z)$. By uniqueness of $z_u^*$, sophisticated consumers therefore receive $y_v^S=(z_u^*,c(z_u^*))$. \\

Naive consumers instead maximize $U(y_u^N)$ subject to the within-contract choice constraints $U(y_u^N)\geq U(y_v^N)$ and $V(y_v^N)\geq V(y_u^N)$, and nonnegative realized profit $p_v^N\geq\int c(z)\,d\alpha_v^N(z)$. As in Lemma~\ref{lemma:hypo}, the anticipated outcome is provided most cheaply by concentrating $\alpha_u^N$ on some $z'\in\arg\max_z\{u(z)-v(z)\}$. The $V$-choice and zero-profit constraints bind, so $p_v^N=\int c(z)\,d\alpha_v^N(z)$ and $p_u^N=v(z')-\int v(z)\,d\alpha_v^N(z)+p_v^N$. Substitution leaves the choice of $\alpha_v^N$ to maximize $\int[v(z)-c(z)]\,d\alpha_v^N(z)$. By uniqueness of $z_v^*$, naive consumers therefore receive $y_v^N=(z_v^*,c(z_v^*))$, supported by the bait outcome $y_u^N=(z',v(z')-v(z_v^*)+c(z_v^*))$.
\end{proof}

\begin{proof}[Proof of Proposition~\ref{prop:comendo}]
Write $R(s)=s-L(s)$. Fix a monopoly optimum $x\in[0,x^C]$. Optimality relative
to every $s\in(x,x^C]$ requires
\[
[1-F(s)]R(s)\le[1-F(x)]R(x),
\]
and therefore, whenever $R(s)>0$,
\[
1-F(s)\le[1-F(x)]\frac{R(x)}{R(s)}
\le\frac{R(x)}{R(s)}.
\]
It follows that
\[
\int_x^{x^C}[1-F(s)]\,ds
\le
\int_x^{x^C}\frac{x-L(x)}{s-L(s)}\,ds.
\]
Thus the stated condition is sufficient for
$CS^C-CS^M=W^*-\int_x^{x^C}[1-F(s)]ds\ge0$ for every $F$.\\

 For necessity, suppose the condition fails at some $x\in(0,x^C]$ with $R(x)>0$. Since $R$ is nondecreasing on $[0,x^C]\subseteq[0,\bar x]$, define \[ 1-F_0(s)=\frac{R(x)}{R(s)}, \qquad s\in[x,x^C], \] and extend $F_0$ by setting $F_0(s)=0$ below $x$ and $F_0(s)=1$ above $x^C$. This makes every $s\in[x,x^C]$ yield the same monopoly objective $R(x)$, and \[ \int_x^{x^C}[1-F_0(s)]\,ds = \int_x^{x^C}\frac{R(x)}{R(s)}\,ds >W^*. \] The survival function can be perturbed arbitrarily little into one generated by a continuous, strictly increasing distribution with full support. Since the inequality is strict, it continues to hold after a sufficiently small perturbation. Hence, competition reduces consumer welfare for some admissible cognitive-cost distribution, establishing necessity.
\end{proof}

\subsubsection*{C.2\quad Sophistication-Promoting Policies}

\begin{proof}[Proof of Proposition~\ref{prop:exoeasp}]
With exogenous cognitive states and sophisticated share $\tau$, consumer
welfare is $ CS(\tau)=-(1-\bar{\eta}(\tau))\bar x $.
It is strictly increasing in $\tau$.
\end{proof}

\begin{proof}[Proof of Proposition~\ref{prop:enoeasp}]
Let $R(x)=x-L(x)$. The monopolist chooses
\[
x^*(\tau)\in\arg\max_{x\in[0,\bar x]}
[1-F(x;\tau)]R(x).
\]
Where the objective is positive, it has the same maximizers as
\begin{equation}\label{eq:aux}
    \log[1-F(x;\tau)]+\log R(x).
\end{equation}
 
By progressivity, for every $x>x'$, the ratio
$[1-F(x;\tau)]/[1-F(x';\tau)]$ is nonincreasing in $\tau$. Taking logs,
this implies that $\log[1-F(x;\tau)]$ has decreasing differences in $(x,\tau)$. The term
$\log R(x)$ is independent of $\tau$, so objective~\eqref{eq:aux} also has
decreasing differences. Monotone comparative statics therefore implies that
the greatest and least selections from the argmax correspondence are
nonincreasing in $\tau$; under our selection, $x^*(\tau)$ is
nonincreasing.\\

Consumer welfare is
\[
CS(\tau)=-\int_0^{x^*(\tau)}[1-F(s;\tau)]\,ds.
\]
Because the policy is effective, $F(s;\tau)$ is nondecreasing in $\tau$, so the survival
function $1-F(s;\tau)$ is nonincreasing. Progressivity makes the upper limit
$x^*(\tau)$ nonincreasing. Both changes weakly reduce the nonnegative integral,
so consumer welfare is nondecreasing in policy intensity.
\end{proof}

\begin{proof}[Proof of Proposition~\ref{prop:expost_result}]

Let $S(\kappa)=1-F(\kappa)$ and $R(x)=x-L(x)$. Under an ex-post policy $\tau$,
the monopolist chooses $x^*(\tau)\in\arg\max_{x\in[0,\bar x]}S((1-\tau)x)R(x) $.
Where positive, the log-objective is $ H(x,\tau)=\log S((1-\tau)x)+\log R(x)$.
Let $h(\kappa)=f(\kappa)/S(\kappa)$ and
$g(\kappa)=\kappa h(\kappa)$. \\

For \(x'>x\), consider \[ H(x',\tau)-H(x,\tau) = \log\frac{S((1-\tau)x')}{S((1-\tau)x)} + \log\frac{R(x')}{R(x)}. \] The second term is independent of \(\tau\). Because \(F\) admits a density, the first term is absolutely continuous in \(\tau\), with derivative almost everywhere given by \[ x'h((1-\tau)x')-xh((1-\tau)x) = \frac{ g((1-\tau)x')-g((1-\tau)x) }{1-\tau}. \] This expression is nonnegative because \(x'>x\) and \(g\) is increasing. Hence, \(H(x',\tau)-H(x,\tau)\) is nondecreasing in \(\tau\), so \(H\) has increasing differences in \((x,\tau)\). Monotone comparative statics therefore implies that the greatest and least optimal exploitation levels are nondecreasing in \(\tau\). Thus, under our selection, \(x^*(\tau)\) is nondecreasing.\\

\end{proof}

\begin{proof}[Proof of Proposition~\ref{prop:elasticity}]
Let \(R(x)=x-L(x)\) and define its elasticity
\[
\mathcal E(x)=\frac{xR'(x)}{R(x)}
=\frac{x[1-L'(x)]}{x-L(x)}
\]
at points where \(R(x)>0\) and \(R'\) exists. Let
\(q\equiv(1-\tau)x\) be a change of variables. Moreover, define
\(S(x)\equiv1-F(x)\), so that the monopolist's problem can be rewritten as
\[
q^*(\tau)\in\arg\max_{q\in[0,(1-\tau)\bar x]}
S(q)R\!\left(\frac{q}{1-\tau}\right).
\]
Consumer welfare is
\[
CS(\tau)=-\int_0^{q^*(\tau)}S(s)\,ds,
\]
so it is nondecreasing whenever \(q^*(\tau)\) is nonincreasing. 

Set \(a=(1-\tau)^{-1}\), which is increasing in \(\tau\). Applying a
log transformation to the monopolist's objective, the log-objective
becomes $\log S(q)+\log R(aq).$
The cross-partial with respect to \(q\) and \(a\) is
\[
\frac{\partial^2\log R(aq)}{\partial a\,\partial q}
=
\frac{1}{aq}
\frac{d\mathcal E(x)}{d\log x}\bigg|_{x=aq}.
\]
Hence, the log-objective has decreasing differences in \((q,a)\)
whenever \(\mathcal E(x)\) is nonincreasing in \(x\). Monotone
comparative statics then implies that the smallest maximizer
\(q^*(\tau)\) is nonincreasing for every survival function \(S\),
proving sufficiency.\\

For necessity, suppose that \(\mathcal E\) is not nonincreasing. By
the cross-partial calculation above, there exist \(q'<q''\) and
\(a'<a''\) such that \(\log R(aq)\) violates decreasing differences
on \([q',q'']\times[a',a'']\):
\[
\log R(a''q'')-\log R(a''q')
>
\log R(a'q'')-\log R(a'q').
\]
Choose \(\varepsilon>0\) smaller than this strict violation and choose
\(S(q')\) and \(S(q'')\) so that
\[
\log S(q'')-\log S(q')
=
-\bigl[\log R(a'q'')-\log R(a'q')\bigr]-\varepsilon.
\]
The monopolist then strictly prefers \(q'\) to \(q''\) at \(a'\).
Moreover,
\begin{align*}
&\log S(q'')+\log R(a''q'')
-\log S(q')-\log R(a''q')\\
&\quad=
\bigl[\log R(a''q'')-\log R(a''q')\bigr]
-\bigl[\log R(a'q'')-\log R(a'q')\bigr]
-\varepsilon
>0,
\end{align*}
so \(q''\) is strictly preferred to \(q'\) at \(a''\). \\

The two prescribed survival probabilities can be extended to a
nonincreasing survival function \(S\) that makes \(q'\) the unique
maximizer at \(a'\) and \(q''\) the unique maximizer at \(a''\). If
this survival function is not continuous, strictly decreasing, and
of full support, it can be approximated arbitrarily closely by
survival functions satisfying these properties: Because both
comparisons are strict, the reversal persists under all sufficiently
close approximations. Hence, there exists a cognitive-cost
distribution \(F=1-S\) in \(\mathcal F\) for which the monopolist's
selected choice increases as policy intensity rises. Since consumer
welfare is strictly decreasing in the selected choice, consumer
welfare falls for this distribution. Consequently, consumer welfare
cannot be nondecreasing for every cognitive-cost distribution
\(F\in\mathcal F\) unless \(\mathcal E\) is nonincreasing.
\end{proof}

\subsubsection*{C.3\quad Preference Nudges}

\begin{proof}[Proof of Proposition~\ref{prop:exoN}]
We first characterize the monopolist's optimal contracts. Under a nudge of
intensity $\tau$, a fraction $\tau$ of consumers is dynamically consistent.
Among the remaining fraction $1-\tau$ of dynamically inconsistent consumers,
a fraction $\eta$ is sophisticated and a fraction $1-\eta$ is naive.

\begin{lemma}\label{lemma:exoN}
Suppose that cognitive states are exogenous. The profit-maximizing menu contains
the sophisticated contract $y^S=(z_u^*,u(z_u^*))$ and a second contract chosen
by both naive dynamically inconsistent consumers and dynamically consistent
consumers. The latter contract induces
\[
y_u^N
=
\left(\alpha_u^N,\int u(z)\,d\alpha_u^N(z)\right)
\quad\text{and}\quad
y_v^N
=
\left(
z_v^*,
v(z_v^*)+\int [u(z)-v(z)]\,d\alpha_u^N(z)
\right),
\]
where $\alpha_u^N$ may be any lottery supported on
\[
Z(\tau)
\equiv
\arg\max_z
\left\{
\tau [u(z)-c(z)]
+
(1-\tau)(1-\eta)[u(z)-v(z)]
\right\}.
\]
Sophisticated dynamically inconsistent consumers choose $y^S$. Naive
dynamically inconsistent consumers and dynamically consistent consumers choose
the second contract, with the former implementing $y_v^N$ and the latter
implementing $y_u^N$.
\end{lemma}

\begin{proof}
We first derive an upper bound by allowing the monopolist to observe whether a
consumer is sophisticated and then show that the resulting contracts remain
implementable when sophistication is unobserved.\\

A sophisticated dynamically inconsistent consumer correctly anticipates her
outcome-stage behavior. The monopolist can therefore offer her the efficient
contract $y^S=(z_u^*,u(z_u^*))$, which implements $z_u^*$ and extracts her
entire welfare. Profit from each sophisticated consumer is consequently
$\pi_S=u(z_u^*)-c(z_u^*)$.\\

It remains to characterize the contract selected by naive dynamically
inconsistent consumers and dynamically consistent consumers. The two types are
behaviorally indistinguishable at the contract-selection stage because both
evaluate contracts according to $u$, and hence select the same contract. At the
outcome stage, however, a dynamically consistent consumer continues to evaluate
outcomes according to $u$ and implements $y_u^N$, whereas a naive dynamically
inconsistent consumer evaluates outcomes according to $v$ and implements
$y_v^N$.\\

Profit from sophisticated consumers is independent of the second contract.
Omitting this constant, the monopolist solves
\[
\max_{\alpha_u^N,\alpha_v^N,p_u^N,p_v^N}
\left\{
\tau\left[p_u^N-\int c(z)\,d\alpha_u^N(z)\right]
+
(1-\tau)(1-\eta)
\left[p_v^N-\int c(z)\,d\alpha_v^N(z)\right]
\right\}
\]
subject to $U(y_u^N)\geq0$, $V(y_v^N)\geq V(y_u^N)$, and
$U(y_u^N)\geq U(y_v^N)$.\\

We ignore (and later verify) the last constraint. The remaining constraints are
$p_u^N\leq\int u(z)\,d\alpha_u^N(z)$ and
\[
p_v^N
\leq
p_u^N
+
\int v(z)\,d\alpha_v^N(z)
-
\int v(z)\,d\alpha_u^N(z).
\]
Both bind at the optimum: Raising $p_u^N$ directly increases profit from
dynamically consistent consumers and relaxes the upper bound on $p_v^N$;
conditional on $p_u^N$, raising $p_v^N$ increases profit from naive consumers.
Thus, $p_u^N=\int u(z)\,d\alpha_u^N(z)$ and
\[
p_v^N
=
\int u(z)\,d\alpha_u^N(z)
+
\int v(z)\,d\alpha_v^N(z)
-
\int v(z)\,d\alpha_u^N(z).
\]

Substitution reduces the problem to
\[
\max_{\alpha_u^N,\alpha_v^N}
\left\{
\tau\int [u(z)-c(z)]\,d\alpha_u^N(z)
+
(1-\tau)(1-\eta)
\left[
\int [u(z)-v(z)]\,d\alpha_u^N(z)
+
\int [v(z)-c(z)]\,d\alpha_v^N(z)
\right]
\right\}.
\]
The two allocations can be optimized separately. Since \(z_v^*\) maximizes \(v(z)-c(z)\), the optimal \(\alpha_v^N\) assigns unit mass to \(z_v^*\). The optimal \(\alpha_u^N\) places all its mass on \(Z(\tau)\). Moreover, every optimal bait allocation satisfies \[ \int B(z)\,d\alpha_u^N(z)\geq B(z_v^*). \] To see this, write \(C(z)\equiv v(z)-c(z)\), so that \(A(z)=B(z)+C(z)\). The bait objective is therefore \[ \bigl[\tau+(1-\tau)(1-\eta)\bigr] \int B(z)\,d\alpha_u^N(z) + \tau\int C(z)\,d\alpha_u^N(z). \] Because \(z_v^*\) maximizes \(C\), any lottery generating a lower expected value of \(B\) than \(B(z_v^*)\) yields a strictly lower objective than \(\delta_{z_v^*}\). Hence, \[ \int B(z)\,d\alpha_u^N(z)\geq B(z_v^*). \] It follows that \(p_u^N=\int u(z)\,d\alpha_u^N(z)\) and \(p_v^N=v(z_v^*)+\int[u(z)-v(z)]\,d\alpha_u^N(z)\), as stated.\\

It remains to verify the omitted constraint. By construction,
$U(y_u^N)=0$, while
\[
U(y_v^N)
=
u(z_v^*)-v(z_v^*)
-
\int [u(z)-v(z)]\,d\alpha_u^N(z)
=
-x(\tau),
\]
where
\[
x(\tau)
\equiv
\int [u(z)-v(z)]\,d\alpha_u^N(z)
-
\bigl[u(z_v^*)-v(z_v^*)\bigr]
\geq0
\]
is the induced exploitation depth. Hence $U(y_u^N)\geq U(y_v^N)$, so the
omitted constraint is satisfied. Finally, both $y^S$ and $y_u^N$ provide zero
utility at the contract-selection stage. The menu therefore remains
implementable when the monopolist cannot observe sophistication.
\end{proof}

We now establish the comparative statics of exploitation. Define
$A(z)\equiv u(z)-c(z)$ and $B(z)\equiv u(z)-v(z)$. By
Lemma~\ref{lemma:exoN}, $\alpha_u^N(\tau)$ is supported on
\[
Z(\tau)
=
\arg\max_z
\left\{
\tau A(z)+(1-\tau)(1-\eta)B(z)
\right\}.
\]
For $\tau<1$, this is equivalent to maximizing
$B(z)+\lambda(\tau)A(z)$, where
$\lambda(\tau)\equiv\tau/[(1-\tau)(1-\eta)]$ is increasing in $\tau$. Thus, a
stronger nudge places greater relative weight on the allocative welfare
generated by the bait outcome and less relative weight on its usefulness for
exploiting naive consumers.\\

Take $\tau''>\tau'$, choose $z'\in Z(\tau')$ and
$z''\in Z(\tau'')$, and write $\lambda'=\lambda(\tau')$ and
$\lambda''=\lambda(\tau'')$. Optimality gives
$B(z')+\lambda'A(z')\geq B(z'')+\lambda'A(z'')$ and
$B(z'')+\lambda''A(z'')\geq B(z')+\lambda''A(z')$. Adding these inequalities
and using $\lambda''>\lambda'$ yields $A(z'')\geq A(z')$; substitution into
the first inequality then gives $B(z'')\leq B(z')$. Therefore,
\[
A_\tau
\equiv
\int A(z)\,d\alpha_u^N(\tau)(z)
\quad\text{is nondecreasing, whereas}\quad
B_\tau
\equiv
\int B(z)\,d\alpha_u^N(\tau)(z)
\quad\text{is nonincreasing.}
\]

The exploitation depth imposed on a naive consumer is
$x(\tau)=B_\tau-B(z_v^*)$. Because $B(z_v^*)$ is independent of $\tau$ and
$B_\tau$ is nonincreasing, $x(\tau)$ is nonincreasing. \\

Dynamically consistent and sophisticated dynamically inconsistent consumers
receive zero utility, whereas each naive dynamically inconsistent consumer
receives utility $-x(\tau)$. Aggregate consumer welfare is therefore
\[
CS(\tau)=-(1-\tau)(1-\eta)x(\tau).
\]
Both $1-\tau$ and $x(\tau)$ are nonnegative and nonincreasing in $\tau$.
Their product is therefore nonincreasing, so consumer welfare is
nondecreasing.
\end{proof}

\begin{proof}[Proof of Proposition~\ref{prop:endoN}]
Let \(S(x)\equiv1-F(x)\), \(D(z)\equiv u(z)-v(z)\), and \(D^{\max}\equiv\max_{z\in Z}D(z)\). For a bait allocation \(\alpha_u\), define \(A(\alpha_u)\equiv\int_0^1[u(z)-c(z)]\,d\alpha_u(z)\) and \(\Delta(\alpha_u)\equiv D^{\max}-\int_0^1D(z)\,d\alpha_u(z)\).\\

By the same arguments used to derive the benchmark contracting problem, conditional on \((x,\alpha_u)\), sophisticated consumers generate profit \(W^*\), dynamically consistent consumers generate profit \(A(\alpha_u)\), and naive consumers generate profit \(W^*+x-L\bigl(x+\Delta(\alpha_u)\bigr)\). Up to the constant \((1-\tau)W^*\), the monopolist's problem is therefore
\[
\max_{x,\alpha_u}
\left\{
\tau A(\alpha_u)
+
(1-\tau)S(x)
\left[x-L\bigl(x+\Delta(\alpha_u)\bigr)\right]
\right\}.
\]

For a fixed \(\Delta\), define \(G(x,\Delta)\equiv S(x)[x-L(x+\Delta)]\), let \(\underline{x}(\Delta)\equiv\min\arg\max_xG(x,\Delta)\), and define \(V(\Delta)\equiv\max_xG(x,\Delta)\). Wherever \(G(x,\Delta)>0\), the problem has the same maximizers as \(\log S(x)+\log[x-L(x+\Delta)]\), and
\[
\frac{\partial}{\partial x}\log\bigl[x-L(x+\Delta)\bigr]
=
\frac{1-L'(x+\Delta)}{x-L(x+\Delta)}.
\]
By the condition in the proposition, this expression is nonincreasing in \(\Delta\). The conditional objective therefore has decreasing differences in \((x,\Delta)\), so monotone comparative statics implies that \(\underline{x}(\Delta)\) is nonincreasing in \(\Delta\). Moreover, because \(L\) is increasing, \(G(x,\Delta)\) is nonincreasing in \(\Delta\) for every \(x\), and hence \(V(\Delta)\) is nonincreasing in \(\Delta\).\\

For \(\tau<1\), divide the monopolist's objective by \(1-\tau\) and let \(\lambda(\tau)\equiv\tau/(1-\tau)\). The remaining problem over bait allocations is
\[
\max_{\alpha_u}
\left\{
V\bigl(\Delta(\alpha_u)\bigr)
+
\lambda(\tau)A(\alpha_u)
\right\}.
\]
Consider \(\tau''>\tau'\), and let selected optimal bait allocations induce \((A'',\Delta'')\) and \((A',\Delta')\), respectively. Because \(\lambda(\tau)\) increases with \(\tau\), standard revealed-preference arguments imply that \(A''\geq A'\) and \(V(\Delta'')\leq V(\Delta')\).\\

If \(\Delta''\geq\Delta'\), the conditional comparative static gives \(\underline{x}(\Delta'')\leq\underline{x}(\Delta')\). If instead \(\Delta''<\Delta'\), monotonicity of \(V\) gives \(V(\Delta'')\geq V(\Delta')\). Hence, \(V(\Delta'')=V(\Delta')\) and \(A''=A'\), so both bait allocations are optimal at both policy intensities. In particular, the bait allocation selected at \(\tau'\) remains optimal at \(\tau''\). Because the equilibrium selection chooses the smallest exploitation depth among all profit-maximizing menus, the selected depth at \(\tau''\) cannot exceed that at \(\tau'\). Thus, in either case, \(x^*(\tau'')\leq x^*(\tau')\). At \(\tau=1\), no consumers are dynamically inconsistent, and the smallest-depth convention selects \(x^*(1)=0\). Therefore, \(x^*(\tau)\) is nonincreasing in \(\tau\).\\

Finally, consumer welfare is
\[
CS(\tau)
=
-(1-\tau)\int_0^{x^*(\tau)}[1-F(s)]\,ds.
\]
Because both \(1-\tau\) and \(x^*(\tau)\) are nonincreasing in \(\tau\), while the integrand is nonnegative, the consumer-welfare loss is nonincreasing in policy intensity. Hence, consumer welfare is nondecreasing in \(\tau\) for every cognitive-cost distribution \(F\).
\end{proof}

\end{document}